\documentclass[journal, draftcls, onecolumn, 12pt]{IEEEtran}

\usepackage{amsmath}
\usepackage{amssymb}
\usepackage{amsfonts}
\usepackage{graphicx}
\usepackage{graphics}
\usepackage{theorem}
\usepackage{euscript}
\usepackage{psfrag}

\def\N{{\mathbb N}}        
\def\R{{\mathbb R}}        
\def\C{{\mathbb C}}        
\def\L{{\mathbb L}}

\def\P{{\mathbb P}}        
\def\E{{\mathbb E}}        
\def\1{{\mathbf 1}}        
\def\ep{\varepsilon}

\DeclareMathOperator{\tr}{tr} \DeclareMathOperator{\var}{Var}
\DeclareMathOperator{\supp}{Supp}

\theoremstyle{plain}
\newtheorem{thm}{Theorem}
\newtheorem{cor}[thm]{Corollary}
\newtheorem{lem}[thm]{Lemma}
\newtheorem{prop}[thm]{Proposition}

\newtheorem{deft}[thm]{Definition}

\newcommand{\ebno}{{E_b/N_0}}

\newcommand{\febno}{\frac{E_b}{N_0}}

\newcommand{\So}{S_0}
\newcommand{\Sinf}{\mathcal{S}_\infty}
\newcommand{\Linf}{\mathcal{L}_\infty}

\newcommand{\vct}[1]{\boldsymbol{#1}}
\newcommand{\Mat}[1]{\boldsymbol{#1}}
\newcommand{\abs}[1]{\left\vert#1\right\vert}

\newcommand{\tendvers}[1]{\xrightarrow[n\rightarrow\infty]{}#1}
\newcommand{\nrm}[1]{\left\lVert#1\right\rVert}

\DeclareMathOperator{\trace}{tr}

\begin{document}

\title{On Certain Large Random Hermitian Jacobi Matrices with Applications to Wireless Communications}

\author{\authorblockN{Nathan Levy\authorrefmark{1}\authorrefmark{2}, Oren Somekh\authorrefmark{3}, Shlomo Shamai (Shitz)\authorrefmark{2}} and Ofer Zeitouni\authorrefmark{4}\authorrefmark{5}\\
\authorblockA{\authorrefmark{1}
D\'epartement de Math\'ematiques et Applications, Ecole Normale Sup\'erieure, Paris 75005,
France}\\
\authorblockA{\authorrefmark{2}
Department of Electrical Engineering, Technion, Haifa 32000,
Israel}\\
\authorblockA{\authorrefmark{3}
Department of Electrical Engineering, Princeton University,
Princeton, NJ 08544, USA}\\
\authorblockA{\authorrefmark{4}
School of Mathematics, University of Minnesota, Minneapolis, MN 55455,
USA}\\
\authorblockA{\authorrefmark{5}
Department of Mathematics, Weizmann Institute of Science, Rehovot 76100, Israel}\\
}

\maketitle

\begin{abstract}
In this paper we study the spectrum of certain large random
Hermitian \emph{Jacobi} matrices. These matrices are known to
describe certain communication setups. In particular we are
interested in an uplink cellular channel which models mobile users
experiencing a soft-handoff situation under joint multicell
decoding. Considering rather general fading statistics we provide
a closed form expression
for the per-cell
sum-rate of this channel
in high-SNR,
when an intra-cell TDMA protocol is
employed. Since the matrices of interest are \emph{tridiagonal},
their eigenvectors can be considered as sequences with second
order linear recurrence. Therefore, the problem is reduced to the
study of the exponential growth of products of two by two
matrices. For the case where $K$ users are simultaneously active
in each cell, we obtain a series of lower and upper bound on the
high-SNR power offset of the per-cell sum-rate, which are
considerably tighter than previously known bounds.
\end{abstract}

\section{Introduction}
The growing demand for ubiquitous access to high-data rate
services, has produced a huge amount of research analyzing the
performance of wireless communications systems. Cellular systems
are of major interest as the most common method for providing
continuous services to mobile users, in both indoor and outdoor
environments. Techniques for
providing better service and
coverage in cellular mobile communications are currently being
investigated by industry and academia. In particular, the use of
joint multi-cell processing (MCP), which allows the base-stations
(BSs) to jointly process their signals, equivalently creating a
distributed antenna array, has been identified as a key tool for
enhancing system performance (see
\cite{Shamai-Somekh-Zaidel-JWCC-2004}\cite{Somekh-Simeone-Barness-Haimovich-Shamai-BookChapt-07}
and references therein for surveys of recent results on multi-cell
processing).

Most of the works on the uplink channel of cellular systems deal
with a single-cell setup. References that consider multi-cell
scenarios tend to adopt complex multi-cell system models which
render analytical treatment extremely hard (if not, impossible).
Indeed, most of the results reported in these works are derived
via intensive numerical calculations which provide little insight
into the behavior of the system performance as a function of
various key parameters (e.g.
\cite{Viterbi-Vehic-1991}-\nocite{Hanly-Whiting-Telc-1993}\nocite{Zorzi-Vehic-1996}\nocite{Zorzi-Comm-1996}\nocite{Rimoldi-Li-1996}\cite{Jafar-Foschini-Goldmisth-2004}).

Motivated by the fact that mobiles users in a cellular system
``see" only a small number of BSs,
and by the desire to provide analytical results, an
attractive analytically tractable model for a multi-cell system
was suggested by Wyner in \cite{Wyner-94} (see also
\cite{Hanly-Whiting-Telc-1993} for an earlier relevant work).
In this model, the system's cells are ordered in either an infinite
linear array, or in the familiar two-dimensional hexagonal pattern
(also infinite). It is assumed that only adjacent-cell
interference is present and characterized by a single parameter, a
scaling factor $\alpha\in[0,1]$. Considering non-fading channels
and a ``wideband'' (WB) transmission scheme, where all bandwidth
is available for coding (as opposed to \emph{random} spreading),
the throughput obtained with optimum and linear MMSE \emph{joint}
processing of the received signals from \emph{all} cell-sites are
derived. Since it was first presented in \cite{Wyner-94},
``Wyner-like" models have provided a framework for many works
analyzing various transmission schemes in both the uplink and
downlink channels (see
\cite{Somekh-Simeone-Barness-Haimovich-Shamai-BookChapt-07} and
references therein).

In this work we consider a simple ``Wyner-like" cellular setup
presented in \cite{Somekh-Zaidel-Shamai-CWIT-05} (see also
\cite{Somekh-Zaidel-Shamai-IT-2005}). According to this setup, the
cells are arranged on a circle (or a line), and the mobile users
``see" only the two BSs which are located on their cell's
boundaries. All the BSs are assumed to be connected through an
ideal backhaul network to a central multi-cell processor (MCP),
that can \emph{jointly} process the uplink received signals of all
cell-sites, as well as pre-process the signals to be transmitted
by all cell-sites in the downlink channel. The users are hence in
what is referred to as a ``soft-handoff" situation, which is very
common in practical real-life cellular systems, and is therefore
of real practical as well as theoretical interest (see for example
\cite{Nasser-Hasswa-Hassanein-Comm-Mag-2006} for a recent survey
on handoff schemes).  With simplicity and analytical tractability
in mind, and in a similar manner to previous
 work, the model
provides perhaps the simplest framework for a soft-handoff setting
in a cellular system, that still represents real-life phenomena
such as intercell interference and fading.

Unfortunately, the analysis of ``Wyner-like" models in general and
the ``soft-handoff" setup in particular presents some analytical
difficulties (see Section \ref{sec: Analysis Difficulty}) when
fading is present. These difficulties render conventional analysis
methods such as large random matrix theory impractical. Indeed the
per-cell sum-rate rates supported by MCP in the uplink channel of
the ``soft-handoff" setups are known only for limited scenarios
such as non-fading channels, phase-fading channels, fading
channels but with large number of users per-cell, and Rayleigh
fading channels with single user active per-cell
\cite{Somekh-Zaidel-Shamai-CWIT-05}\cite{Somekh-Zaidel-Shamai-IT-2005}\cite{Jing-Tse-Hou-Soriaga-Smee-Padovani-ISIT-2007}.
The latter result is due to a remarkable early work by Narula
\cite{Narula-1997} dealing with the capacity of a two-tap time
variant ISI channel. Calculating the per-cell sum-rate capacity
supported by the uplink channel of the ``soft-handoff" setup in
the presence of \emph{general} fading channels (not necessarily
Rayleigh fading channels), when \emph{finite} number of users are
active simultaneously in each cell remains an open problem (see
\cite{Somekh-Zaidel-Shamai-IT-2005}\cite{Lifang-Goldsmith-Globecom06}
for bounds on this rate). As will be shown in the sequel, this
problem is closely related to calculating the spectrum of certain
large random Hermitian Jacobi matrices.
The high-SNR
characterization of the sum-rate capacity, previously unknown, is
the main focus of this work.

In particular we calculate the high-SNR slope and power offset of
the rate with a single user active per-cell (intra-cell TDMA)
under a rather generic fading distribution. We also prove the
following results for any given number of active users per-cell.
We prove the existence of a limiting sum-rate capacity when the number of
cells goes to infinity and calculate the high-SNR slope in Theorem
\ref{principalKtext}. Moreover, we give bounds on the high-SNR
power offset in Proposition \ref{boundsK}. In particular, we give
a sequence of explicit upper- and lower-bounds; the gap between
the lower and the upper bounds is decreasing with the bounds'
order and complexity.

The rest of the paper is organized as follows. In Section
\ref{sec: Problem Statement} we present the problem statement and
main results. Section \ref{sec: Background} includes a
comprehensive
review of previous works. Several applications
of the main result are discussed in Section \ref{sec:
Applications}. Concluding remarks are included in Section
\ref{sec: Concluding Remarks}. Various derivations and proofs
are deferred to the Appendices.

\section{Problem Statement and Main Results}\label{sec: Problem Statement}

\subsection{System Model}
In this paper we consider a linear version of the cellular
``soft-handoff" setup introduced in
\cite{Somekh-Zaidel-Shamai-CWIT-05}\cite{Somekh-Zaidel-Shamai-IT-2005},
according to which $M+1$ cells with $K$ single antenna users per cell are
arranged on a line, where the $M$ single antenna BSs are located
on the boundaries of the cells (see Fig. \ref{fig: Soft-handoff
setup} for the special case of $M=3$). Starting with the WB
transmission scheme where all bandwidth is devoted for coding and
all $K$ users are transmitting simultaneously each with average
power $\rho$, and assuming synchronized communication, a vector
baseband representation of the signals received at the system's
BSs is given for an arbitrary time index by
\begin{equation}\label{eq: Cahnnel definition}
    \vct{y}=\Mat{H_M}\vct{x}+\vct{n}\ .
\end{equation}
The $M\times K(M+1)$ channel transfer matrix $\Mat{H}_M$ is a two
block diagonal matrix defined by
\begin{equation}\label{eq: General channel transfer matrix}
\Mat{H}_{M}=\left(
\begin{array}{ccccc}
\vct{a}_1   & \vct{b}_1     & \vct{0}         & \cdots            & \vct{0}       \\
\vct{0}     & \ddots        & \ddots          & \ddots            & \vdots        \\
\vdots      & \ddots        & \ddots          & \ddots            & \vct{0}       \\
\vct{0}     & \cdots        & \vct{0}         & \vct{a}_{M}       & \vct{b}_{M}   \\
\end{array}
\right)\ ,
\end{equation}
where $\vct{a}_m$ and $\vct{b}_m$ are $1\times K$ row vectors
denoting the channel complex fading coefficients, experienced by
the $K$ users of the $m$th and $(m+1)$th cells, respectively, when
received by the $m$th BS antenna.
$\vct{n}$
represents the $M\times M$ zero mean circularly symmetric Gaussian
noise vector $\vct{n}\sim \mathcal{CN}(\vct{0},\Mat{I}_M)$.

We assume throughout
that the fading processes are i.i.d.\ among different users and BSs,
with
$a_{m,k}\sim \pi_a$ and $b_{m,k}\sim \pi_b$, and can be viewed for
each user as ergodic processes with respect to the time index. We
denote by $\P$ the probability associated with those random
sequences and by $\E$ the associated expectation. We will be
working throughout with a subset of the following assumptions.
\begin{enumerate}
    \item[(H1)]\label{integrability}
        $\E_{\pi_a}\left(\log\abs{x}\right)^2<\infty$\footnote{A natural base logarithm is used throughout
this work unless explicitly denoted otherwise.} and
$\E_{\pi_b}\left(\log\abs{x}\right)^2<\infty$.
    \item[(H2)]\label{ac} $\pi_a$ and $\pi_b$
        are absolutely continuous with respect to Lebesgue measure on $\C$.
    \item[(H3)]\label{support} There exists a real
        ${\cal M}$ such that if $x$ is distributed
        according to $\pi_a$ (resp. $\pi_b$)
        then the density of $\abs{x}^2$ is
        strictly positive on the interval $[{\cal M};\infty)$.
  \item[(H3')] There exist $m_a<{\cal M}_a\in\R^+\cup\{\infty\}$
      (resp. $m_b<{\cal M}_b\in\R^+\cup\{\infty\}$)
      such that if $x$ is distributed according to
      $\pi_a$ (resp. $\pi_b$) then the density of
      $\abs{x}^2$ and the Lebesgue-measure on $[m_a;{\cal M}_a]$
      (resp. $[m_b;{\cal M}_b]$) are mutually absolutely continuous.
  \item[(H4)] There exists a ball in $\C$ such that the Lebesgue measure outside that ball is absolutely continuous with respect to $\pi_a$ and $\pi_b$.
\end{enumerate}

We further assume that the channel state information (CSI) is
available to the MCP only, while the transmitters know only the
channel statistics, and cannot cooperate their transmissions in
any way. Therefore, independent zero mean circularly symmetric
Gaussian codebooks conform with the capacity achieving statistics,
where $\vct{x}$ denotes the $(M+1)K\times 1$ transmit vector
$\vct{x}\sim \mathcal{CN}(\vct{0},\rho\Mat{I}_{MK})$, and $\rho$
is the average transmit power of each user
\footnote{Note that since the channel transfer matrix
$\Mat{H}_{M}$ is a \emph{column-regular} gain matrix (see
definition in \cite{Tulino-Lozano-Verdu-Ant-Corr-IT-2005}) when
$M\rightarrow\infty$, the capacity achieving statistics remains
the same in this case, even if we allow the users to cooperate as
long as they are unaware of the CSI.}
($\rho$ is thus equal to the transmit SNR of the users).

With the above assumptions, the system \eqref{eq:
Cahnnel definition} is a multiple access channel (MAC). We are
interested in the per-cell sum-rate capacity
\begin{equation}\label{eq: Rate definition}
    C_M(P) =\frac{1}{M}
    \E\left(\log\det{\Mat{G}_M}\right)\quad[\mathrm{nats/channel\ use}]\
    ,
\end{equation}
where $P\triangleq K\rho$ is the per-cell transmitted average
power,
\begin{equation}\label{eq: G def}
    \Mat{G}_M\triangleq\Mat{I}_M+\rho\Mat{H}_M\Mat{H}_M^\dagger\ ,
\end{equation}
and the expectation is taken over the channel transfer matrix
entries.
(Here and in the sequel, for a scalar $z\in \C$, $z^\dagger$ denotes
the complex conjugate, while for a matrix $A$, $A^\dagger$ denotes the
matrix with $A^\dagger(i,j)=A(j,i)^\dagger$.)
The non-zero entries of the \emph{Hermitian Jacobi}
matrix $\Mat{G}_M$ are equal to
\begin{equation}\label{eq: explicit sums - no fading}
\begin{aligned}
\left[\Mat{G}_M\right]_{m,m-1}&=\rho<\vct{b}_{m-1};\vct{a}_m>\,, \\
\left[\Mat{G}_M\right]_{m,m}&=1+\rho \left(\abs{\vct{a}_{m}}^2 +
 \abs{\vct{b}_{m}}^2\right)\,,\\
\left[\Mat{G}_M\right]_{m,m+1}&=\rho<\vct{a}_{m+1};\vct{b}_m>\,,
\end{aligned} \
\end{equation}
where out-of-range indices should be ignored, and for any two
arbitrary $L$ length vectors $\vct{a},\ \vct{b}$ we define
$<\vct{a};\vct{b}> \triangleq \sum_{l=1}^L a^\dagger_{l} b_{l}$,
and $\abs{\vct{a}}^2\triangleq <\vct{a};\vct{a}>$.

Since we shall focus on the asymptotes of infinite number of cells
$M\rightarrow\infty$, boundary effects can be neglected and
symmetry implies that the rate \eqref{eq: Rate definition} equals
the maximum equal rate (or symmetric capacity) supported by the
channel \cite{HW}.

The above description relates to the WB protocol where all users
transmit simultaneously. According to the intra-cell TDMA protocol
only one user is simultaneously active per-cell, transmitting
$1/K$ of the time using the \emph{total} cell transmit power $P$.
In this case it is easily verified that with no loss of
generality, we can consider a single user per cell in terms of the
per-cell sum-rate, setting $K=1$ in \eqref{eq: Cahnnel definition}
and \eqref{eq: General channel transfer matrix}.

\subsection{Analysis Difficulty}\label{sec: Analysis Difficulty}

Many recent studies have analyzed the rates of various channels
using results from (large) random matrix theory (see
\cite{Tulino-Verdu-Random-Matrix-Review-2004} for a recent
review). In those cases, the number of random variables involved is
of the order of the number of elements in the matrix $\Mat{G}_M$
(or $\Mat{H}_M$), and self-averaging is strong enough to ensure
convergence of the empirical measure of eigenvalues, and to derive
equations for the limit (or its Stieltjes transform). In
particular, this is the case if the normalized continuous power
profile of $\Mat{H}_M$, which is defined as
\begin{equation}  \label{eq: MIMO definition of the power profile}
\mathcal{P} _M(r,t)\triangleq \E(\abs{[\Mat{H}_M]_{i,j}}^2) \quad
; \quad \frac{i}{M} \le r < \frac{i+1}{M}\ , \ \frac{j}{(M+1)K}
\le t < \frac{j+1}{(M+1)K}\ ,
\end{equation}
converges uniformly to a bounded, piecewise continuous function as
$ M\to\infty$, see e.g. \cite[Theorem
2.50]{Tulino-Verdu-Random-Matrix-Review-2004} and
\cite{anderson-zeitouni} for fluctuation results. In the case
under consideration here, it is easy to verify that for $K$ fixed,
$\mathcal{P} _M(r,t)$ does \emph{not} converge uniformly, and
other techniques are required.

\subsection{Extreme SNR Regime Characterization}

As mentioned earlier, the per-cell sum-rate capacity of the
``soft-handoff" setup is known only for certain limited
cases to be elaborated in the next section, and in general analytical results are hard to derive. As an
alternative to deriving exact analytical results we focus here on
extracting parameters which characterize the channel rate under
extreme SNR scenarios. The reader is referred to
\cite{Shamai-Verdu-2001-fading}\ -
\nocite{Verdu-paper-low-snr-regime-02}
\cite{Lozano-Tulino-Verdu-high-SNR-IT05} for an elaboration on the
extreme SNR characterization.

\paragraph{The Low-SNR Regime}
This regime is usually the operating
regime for wide-band systems \cite{Verdu-paper-low-snr-regime-02}.

The average per-cell spectral efficiency in bits/sec/Hz, expressed
as a function of the system average transmit SNR, $\ebno$, is
evaluated by solving the implicit equation obtained by
substituting
\begin{equation}\label{}
    P=\mathsf{C}_M\left(\febno\right)\, \febno\
\end{equation}
in \eqref{eq: Rate definition}, where
$\mathsf{C}_M(\ebno)=C_M(P)/\log 2$\ stands for the uplink
spectral efficiency measured in [bits/sec/Hz]. The low-SNR regime
is characterized through the minimum transmit $\ebno$ that enables
reliable communications,
\begin{equation}\label{eq: EbNo-min Definition}
\febno_{\min}\triangleq\frac{\log 2}{\dot{C}_M(0)} \ ,
\end{equation}
and the low-SNR spectral efficiency slope
\begin{equation}\label{eq: So Definition}
\So\triangleq \frac{2\left[\dot{C}_M(0)\right]^2}{-\ddot{C}_M(0)}\
,
\end{equation}
yielding the following low-SNR affine approximation
\begin{equation}\label{eq: Rate lo-snr}
\mathsf{C}_M\left(\febno\right)\approx \frac{\So}{3|_\text{dB}}
\,\left(\left.\febno\right|_\text{dB}-
\left.\febno_{\min}\right|_\text{dB}\right)\quad
[\mathrm{bits/sec/Hz}] .
\end{equation}
In the above definitions $3|_\text{dB}=10\log_{10}2$, and
$\dot{C}_M(0)$ and $\ddot{C}_M(0)$ are the first and second
derivatives (whenever exist) with respect to $P$ of the per-cell
sum-rate capacity, respectively, evaluated at $P=0$. Focusing on
Gaussian channels with receiver CSI only, it can be shown
\cite{Verdu-paper-low-snr-regime-02} that there is no need to
calculate the two derivatives of the rate in $P=0$, and that the
low-SNR parameters are simply given by
\begin{equation}\label{eq: low SNR mimo}
\febno_{\min} = \frac{MK \log
2}{\trace\left(\E\Mat{H}_M^{\dagger}\Mat{H}_M\right)}\quad
;\quad S_0 =
\frac{2}{M}\frac{\left(\trace\left(\E\Mat{H}_M^{\dagger}\Mat{H}_M\right)\right)^2}
{\trace\left(\E\left(\Mat{H}_M^{\dagger}\Mat{H}_M\right)^2\right)}\
.
\end{equation}

\paragraph{The High-SNR Regime}
This is usually the operating
regime for high-data rate (high spectral efficiency) systems
(that is the case actually in all 2.5/3 G standards).

The high-SNR regime is characterized through the high-SNR slope
(also referred to as the ``multiplexing gain", or ``pre-log")
\begin{equation}\label{eq: Sinf Definition}
\Sinf\triangleq \lim_{P\rightarrow\infty} \frac{C_M(P)}{\log
P}=\lim_{P\rightarrow\infty} P \dot{C}_M(P)\ ,
\end{equation}
and the high-SNR power offset
\begin{equation}\label{eq: Linf Definition}
\Linf \triangleq \lim_{P\rightarrow \infty} \frac{1}{\log
2}\left(\log P -\frac{C_M(P)}{\Sinf}\right)\ ,
\end{equation}
yielding the following affine capacity approximation
\begin{equation}\label{}
C_M(P) \approx
  \frac{\Sinf\log 2}{3|_\text{dB}}\left(P|_\text{dB} - 3|_\text{dB}\Linf\right)\ .
\end{equation}
Note that the high-SNR approximation reference channel here is
that of a single isolated cell, with no fading, and total average
transmit power $P$.

The high-SNR characterization of the per-cell sum-rate supported
by the ``soft-handoff" uplink channel is known only in certain
limited scenarios (see Section \ref{sec: Background}) and is the
main focus of this work.

\renewcommand{\labelenumi}{\alph{enumi})}

\subsection{Main Results}\label{sec: Main results}

Recall the definition of $C_M(P)$, c.f.
(\ref{eq: Rate definition}).
Starting with intra-cell TDMA scheme where only one user is active
per-cell transmitting with power $P$ we have the following.
\begin{thm}
\label{principaltext}[intra-cell TDMA scheme $K=1$, high-SNR
characterization] Assume (H1) and (H2)\,.
\begin{enumerate}
    \item For every $P>0$, $C_M(P)$ converges as
        $M$ goes to infinity. We call the limit $C(P)$.
    \item We get the following bounds on $C(P)$,
    \[\max(\E_{\pi_a}\log(1+P\abs{x}^2), \E_{\pi_b}\log(1+P\abs{y}^2))\leq C(P)\leq\E_{\pi_a,\pi_b}\log(1+P(\abs{x}^2+\abs{y}^2)).\]
    \item Further assume [(H3) or
        (H3')]. As $P$ goes to infinity,
        \[C(P)=\log P+
        2\max\left(\E_{\pi_a}\log\abs{x},\E_{\pi_b}\log\abs{x}\right)+o(1).\]
        In particular, ${\cal S}_\infty=1$ and ${\cal L}_\infty=
        -2\max\left(\E_{\pi_a}\log_2\abs{x},\E_{\pi_b}\log_2\abs{x}\right)$.
\end{enumerate}
\end{thm}
Note that point c) shows that the lower bound of point b) is tight
in the high-SNR regime.
\begin{proof}
The proof of points a) and c) follows from Theorem \ref{principal} of Appendix
\ref{sec:DefinitionsAndMainResult}, where we prove that the
variable $\mathcal{C}_M(P)\triangleq 1/M\log\det\Mat{G}_M$
converges almost surely. Note however that
\begin{equation}
\label{hadamard}
0\leq\frac{1}{M}\log\det\Mat{G}_M\leq\frac{1}{M}\sum_{m=1}^M\log\left(1+
\rho(\abs{a_m}^2+\abs{b_m}^2)\right),
\end{equation}
and the second inequality is due to Hadamard's inequality for
semi-positive definite (SPD) hermitian matrices.
With (H1), it follows that $\mathcal{C}_M(P)$ is uniformly
integrable, and hence the almost sure convergence implies
convergence in expectation. Recalling that $C_M(P)=\E\
\mathcal{C}_M(P)$ completes the proof of point a) and c).

Let us show point b) using the tools of \cite{SOW}. We first show the lower bound.  We consider $\vct{n}$, $\vct{x}$ and $\vct{y}$ as in (\ref{eq: Cahnnel definition}).
\begin{align*}
\mathcal{C}_M(P)&=\frac{1}{M}I\left(\vct{x};\vct{y}|(a_i)_{1\leq i\leq M},(b_i)_{1\leq i\leq M}\right)\\
&=\frac{1}{M}\sum_{j=1}^M I\left(x_j;\vct{y}|(x_i)_{1\leq i<j},(a_i)_{1\leq i\leq M},(b_i)_{1\leq i\leq M}\right)\\
&\geq\frac{1}{M}\sum_{j=1}^M I\left(x_j;y_{j-1}|(x_i)_{1\leq i<j},(a_i)_{1\leq i\leq M},(b_i)_{1\leq i\leq M}\right)\\
&=\frac{1}{M}\sum_{j=1}^M I(x_j;b_{j-1}x_j+n_{j-1}|b_{j-1}),
\end{align*}
which is the per-cell sum-rate capacity of a single user fading
channel. Therefore, the lower bound is
\cite{Biglieri-Proakis-Shamai-IEEE-T-IT-98}
$\E_{\pi_b}\log(1+P\abs{y}^2)$. As argued in the proof of Theorem
\ref{principal} in Appendix \ref{sec:DefinitionsAndMainResult}, we
can exchange the role of $\pi_a$ and $\pi_b$, thereby getting the
claimed lower bound.
Finally, the upper bound of b) follows immediately from Hadamard's
inequality for SPD hermitian matrices.
\end{proof}

In the proof of Theorem \ref{principal} (intra-cell TDMA scheme),
we use ideas from the theory of product of random matrices.
Note that $\mathcal{C}_M(P)=1/M \sum_{m=1}^M\log(1+P\lambda_m)$
where $\{\lambda_m\}_{m=1}^M$ are the eigenvalues of
$\Mat{H}_M\Mat{H}_M^\dagger$, and the analysis of capacity hinges
upon the study of spectral properties of
$\Mat{H}_M\Mat{H}_M^\dagger$. The main idea is to link the
spectral properties of the latter matrix with the exponential
growth of the elements of its eigenvectors. Since
$\Mat{H}_M\Mat{H}_M^\dagger$ is a \emph{Hermitian Jacobi} matrix,
hence tridiagonal, its eigenvectors can be considered as sequences
with second order linear recurrence. Therefore, the problem boils
down to the study of the exponential growth of products of two by
two matrices. This is closely related to the evaluation of the top
Lyapunov exponent of the product; The explicit link between
$\mathcal{C}_M(P)$ and the
top Lyapunov exponent is the Thouless
formula (see \cite{enorme} or \cite{groslivre}), a version of
which we prove in Appendix \ref{sec:ProductOfRandomMatrices}. We
emphasize however that we do not use the Thouless formula or
Lyapunov exponents explicitly in the proof of Theorem
\ref{principal}.

Like in the result of Narula \cite{Narula-1997} described below in
Section \ref{sec: Background}, our approach uses the analysis of a
certain Markov Chain. Unlike \cite{Narula-1997}, we are not able
to explicitly evaluate the invariant measure of this chain.
Instead, we use the theory of Harris chains to both prove
convergence and continuity results for the chain. The appropriate
definitions are introduced in the course of proving Theorem
\ref{principal}.

We remark that Theorem \ref{principaltext} continues to hold in a
real setup, that is if instead of (H2), we assume
\begin{enumerate}
  \item[(H2')] $\pi_a$ and $\pi_b$ are supported on $\R$ and are absolutely continuous with respect to Lebesgue measure on $\R$.
\end{enumerate}
Since the argument is identical, we do not discuss this case
further. It is also noted that unlike the non-fading case, where
intra-cell TDMA scheme is optimal (see \cite{Wyner-94}), it is
proved to be suboptimal for $K>1$ in the presence of fading
\cite{Somekh-Shamai-2000}, yet TDMA it is one of the
most common access protocols in cellular systems.

Turning to the WB scheme (which is the capacity achieving scheme
\cite{Somekh-Shamai-2000}), where all the bandwidth is used for
coding, and all $K$ users are transmitting simultaneously with
average power $\rho$ (and total cell average power $P=K \rho$), we
have the following less explicit high-SNR characterization.
\begin{thm}
\label{principalKtext}[WB scheme $K>1$, high-SNR characterization]
Assume (H1), (H2) and (H4), and $K> 1$.
\begin{enumerate}
    \item For every $P>0$, $C_M(P)$ converges
        as $M$ goes to infinity. We call the limit $C(P)$.
    \item We get the following bounds on $C(P)$,
    \[\max(\E\log(1+P\abs{\vct{a}}^2/K), \E\log(1+P\abs{\vct{b}}^2/K))\leq C(P)\leq\E\log(1+P(\abs{\vct{a}}^2+\abs{\vct{b}}^2)/K),\]
    where the expectation is taken in the following way:\
    the random variables $\vct{a}$ and $\vct{b}$ are independent, and
    $\vct{a}$ (resp. $\vct{b}$) is a complex $K$-vector whose coefficients
    are independent and distributed according to $\pi_a$ (resp. $\pi_b$).
    \item As $P$ goes to infinity,
    \begin{equation}\label{eq: Capacity WB hi SNR}
    C(P)=\log P+\E\log\left(\frac{e+\abs{\vct{b}}^2}{K}\right)+o(1)\ ,
\end{equation}
    where the expectation is taken in the following way:\
    the random variables $e$ and $\vct{b}$ are independent, and
    $\vct{b}$ is a complex $K$-vector whose coefficients
    are independent and distributed according to $\pi_b$.
    The law of $e$ is $m_0$, which is the unique invariant
    probability of the Markov chain defined by
    \begin{equation}
        \label{eq-130707a}e_{n+1}=\abs{\vct{a}_n}^2\left(\frac{e_n+
    \abs{\vct{b}_{n-1}}^2\sin^2(\vct{a}_n,\vct{b}_{n-1})}{e_n+\abs{\vct{b}_{n-1}}^2}\right)\
    ,
    \end{equation}
    where for any two arbitrary equal length vectors $\vct{a},\ \vct{b},$
    \begin{equation}\label{eq: sin def}
    \sin^2(\vct{a},\vct{b})\triangleq1-\frac{\abs{<\vct{a};\vct{b}>}^2}{\abs{\vct{a}}^2\abs{\vct{b}}^2}\
    .
\end{equation}
In particular, ${\cal S}_\infty=1$ and ${\cal L}_\infty=
    -\E\log_2\left(\frac{e+\abs{\vct{b}}^2}{K}\right)$.
\end{enumerate}
\end{thm}
As with the case $K=1$, point a) and c) of Theorem
\ref{principalKtext} follow from the almost sure convergence
stated in Theorem \ref{principalK} of Appendix
\ref{sec:KUsersModel}, using (H1) and (\ref{hadamard}). As with
Theorem \ref{principal}, we do not use the Thouless formula or
Lyapunov exponents explicitly in the proof of Theorem
\ref{principalK}. The proof of point b) is the same as the proof
of Theorem \ref{principaltext}.b).
It is worth mentioning that in contrast to Theorem
\ref{principaltext}, the non-asymptotic lower bound b) is not
tight in general for large SNR. This is since it is an increasing
function of $K$ and converges to a rate of a single-user Gaussian
scalar channel, which is smaller than the asymptotic rate of
\eqref{eq: WB large K rate}.

Note that although the roles of the sequences $\{\vct{a}_n\}$ and
$\{\vct{b}_n\}$ in \eqref{eq-130707a} are not symmetric, the
expression \eqref{eq: Capacity WB hi SNR} is symmetric in $\pi_a$
and $\pi_b$, as is the case for $K=1$.

We conclude this section by noting that while Theorem
\ref{principalKtext} (WB scheme $K>1$) does not give explicit
expressions for the high-SNR power offset as Theorem
\ref{principaltext}, its proof leads immediately to easily
computable bounds. In the following, the notation is as in Theorem
\ref{principalKtext}, and we let $e_n(a)$ denote the Markov chain
(\ref{eq-130707a}), with initial condition $e_0(a)=a$.
\begin{prop}
\label{boundsK} Assume (H1), (H2) and (H4), and $K> 1$. Then,
\[\E\log\left(\frac{e_n(0)+\abs{\vct{b}}^2}{K}\right)
\leq\lim_{P\rightarrow\infty}\left[C(P)-\log P\right]
\leq\E\log\left(\frac{e_n(\infty)+\abs{\vct{b}}^2}{K}\right),\]
where the expectation is taken in the following way. $e_n(0)$
(resp. $e_n(\infty)$) and $\vct{b}$ are independent. $\vct{b}$ is
a complex $K$-vector whose coefficients are independent and
distributed according to $\pi_b$. $e_n(0)$ (resp. $e_n(\infty)$)
is the $n$-th step of the Markov chain defined by
(\ref{eq-130707a}) with initial condition $e_0(0)=0$ (resp.
$e_0(\infty)=\infty$).
\end{prop}
Indeed, since the expression (\ref{eq-130707a}) for $e_{n+1}$ is
monotone increasing in $e_n$, the law of $e$ in Theorem
\ref{principalKtext} is stochastically dominated below by the law
of $e_n$ with intial condition $0$, and stochastically dominated
above by the law of $e_n$ with initial condition $\infty$. That
same monotonicity also shows that the sequences of laws of
$e_n(0)$ (resp., $e_n(\infty)$) are monotone increasing (resp.,
decreasing) with respect to stochastic order.

As a direct consequence of Proposition \ref{boundsK} with $n=1$
and \eqref{eq: Linf Definition}, we get the following bounds on
the high-SNR power offset
\begin{align}
\label{bound1}
-\E\log_2\left(\frac{\abs{\vct{a}}^2+\abs{\vct{b}}^2}{K}\right)
\leq\Linf\leq-\E\log_2\left(\frac{\abs{\vct{a}}^2\sin^2(\vct{a},\vct{b})+\abs{\vct{b}'}^2}{K}\right),
\end{align}
where the expectation is taken in the following way: $\vct{a}$,
$\vct{b}$ and $\vct{b}'$ are independent, and $\vct{a}$ (resp.
$\vct{b}$,\ $\vct{b}'$) is a complex $K$-vector whose coefficients
are independent and distributed according to $\pi_a$ (resp.
$\pi_b$).
Note that for $K$ going to infinity, if we assume $\pi_a=\pi_b$
and zero mean, then $\sin^2(\vct{a},\vct{b})$ converges to 1,
therefore the ratio between the upper- and lower-bound of (19),
converges to 1, which also agrees with the asymptotic result of
\eqref{eq: WB extreme SNR general fading}.

\subsubsection*{Numerical Results}\label{sec: Numerical results}

In Figures \ref{fig: Linf order K=2} and \ref{fig: Linf order
K=10} we present the high-SNR power offset bounds of Proposition
\ref{boundsK} in the special case of Rayleigh fading (real and
imaginary parts are independent Gaussian random variables with
zero mean and variance $1/\sqrt{2}$), for $K=2$ and $K=10$ users
per-cell respectively. The curves are produced by Monte Carlo
simulation with $10^5$ samples. The figures
include also the
lower bound of
\cite{Somekh-Zaidel-Shamai-IT-2005}, see
 \eqref{eq: WB extreme SNR Rayleigh
fading}, and the asymptotic results (and lower bound) for large
number of users per-cell $\Linf = -1$ (achieved by taking $K$ to
infinity in \eqref{eq: WB extreme SNR Rayleigh fading}). Examining
the figures it is observed that the new bounds are getting tighter
with
their order $n$ and that the new lower bound is  tighter
than \eqref{eq: WB extreme SNR Rayleigh fading}
already for $n=2$. Moreover, fixing the order $n$, the new bounds
are getting tighter with the number of users per-cell $K$. This
observation is also evident from Fig. \ref{fig: Linf users n=2},
where the bounds are plotted for a fixed order $n=2$ versus the
number of users per-cell $K$. Finally,
since the upper
bound of Fig. \ref{fig: Linf order K=2}
is negative,
    we conclude that the presence of Rayleigh fading is beneficial over
    non-fading channels in the high-SNR region already for $K=2$.
 (See
  \cite{Somekh-Zaidel-Shamai-IT-2005}
 for a similar conclusion in the low-SNR region.)

\section{Background, Previous Results and Bounds}
\label{sec: Background}

In this section we briefly summarize previous work on the
``soft-handoff" uplink cellular model introduced in
\cite{Somekh-Zaidel-Shamai-CWIT-05}\cite{Somekh-Zaidel-Shamai-IT-2005}.
For conciseness,
 we restrict the discussion to the case
where $\pi_a=\pi_b$. Most of the results in the sequel can be
extended to include the general case where $\pi_a\neq\pi_b$.

Starting with non-fading channels (i.e., when $\pi_a$ and $\pi_b$
are singletons at 1), the per-cell sum-rate capacity of the uplink
channel is given for $M\rightarrow\infty$ by
\cite{Somekh-Zaidel-Shamai-IT-2005}
\begin{equation}\label{eq: Rate no
fading} R_{\mathrm{nf}}=\log \left( \frac{1+2
P+\sqrt{1+4P}}{2}\right)\ .
\end{equation}
This rate is achieved by any symmetric intra-cell protocol with
average transmit power of $P$ (e.g. intra-cell TDMA, and WB
protocols). It is noted that the same result holds also for phase
fading processes
\cite{Jing-Tse-Hou-Soriaga-Smee-Padovani-ISIT-2007}.

The extreme SNR characterization of \eqref{eq: Rate no fading} is
summarized for the non-fading setup by
\begin{equation}\label{eq: DL SH no fading extreme SNR}
S_{0}=\frac{4}{3}\,, \quad \frac{E_{b}}{N_{0}}_{\min }=\frac{\log
2}{2}\,,\quad
S_{\infty }=1\,,\quad
\mathcal{L}_{\infty }=0\,.
\end{equation}

Returning to the flat fading setup, the channel coefficients are
taken as i.i.d.\  random variables, denoting by
\begin{equation}\label{eq: General Fading - Moments Notation}
\begin{aligned}
   m_1 &\triangleq \E(a_{m,k})=\E(b_{m,k}) \quad ;
  \quad
   m_2 \triangleq \E(\abs{a_{m,k}}^2)=\E(\abs{b_{m,k}}^2) \, \quad
   \\
  m_4 &\triangleq \E(\abs{a_{m,k}}^4)=\E(\abs{b_{m,k}}^4)  \quad ; \quad \mathcal{K}\triangleq
  \frac{m_4}{m_2^2}
\end{aligned} \quad , \ \forall\ m,k
\end{equation}
the mean, second power moment, fourth power moment and the
{kurtosis} of an individual fading coefficient.

The per-cell sum-rate capacity of the WB scheme with fixed $P$ and
increasing number of users and cells $M,K\rightarrow\infty$, is
given by \cite{Somekh-Zaidel-Shamai-IT-2005}\footnote{Here,
the number of users $K$ is taken to infinity and then
the number of cells $M$ is taken to infinity.}
\begin{equation}\label{eq: WB large K rate}
    R_{\mathrm{wb-f}}=
    \log\left(\frac{1+2P m_2 + \sqrt{1+4P m_2+4P^2(m_2^2-\abs{m_1}^4})}{2} \right)
    \ .
\end{equation}
 The rate is maximized for a zero mean fading distribution
and is given by
\begin{equation}\label{eq: WB large K rate zero mean}
R_{\mathrm{wb-f}}=\log (1+2m_2P )\ .
\end{equation}
Comparing \eqref{eq: Rate no fading} and \eqref{eq: WB large K
rate zero mean} (with $m_2=1$),
it follows that the presence of
fading is beneficial in case the number of users is large.
We note that \eqref{eq: WB large K rate} is also shown in
\cite{Somekh-Zaidel-Shamai-IT-2005} to upper bound the respective
rate for any finite number of users $K$.

Returning to the intra-cell TDMA ($K=1$), for which standard
random matrix theory is not suitable (see Sec. \ref{sec: Analysis
Difficulty}), the  powerful moment bounding technique employed in
\cite{Somekh-Shamai-2000} for the Wyner model, can be utilized to
obtain lower and upper bounds on the per-cell sum-rate.

An alternative approach which replaces the role of the singular
values with the diagonal elements of the \emph{Cholesky}
decomposition of the the matrix $\Mat{G}_M$, was presented by
Narula \cite{Narula-1997} for a two diagonal nonzero channel
matrix $\Mat{H}_M$ whose entries are i.i.d. zero-mean complex
Gaussian (Rayleigh fading). Originally, Narula had studied the
capacity of a time varying
two taps inter-symbol-interference
(ISI) channel, where the channel coefficients are i.i.d. zero-mean
complex Gaussian. With the above assumptions regarding the ISI
channel coefficients it is easy to verify that the capacity of
this model is equal to the per-cell sum-rate capacity of an uplink
intra-cell TDMA scheme employed in the ``soft-handoff'' model.

Following \cite{Narula-1997}, we use the
\emph{Cholesky} decomposition applied to the covariance matrix of
the uplink intra-cell TDMA scheme output vector $\Mat{G}_M
=\boldsymbol{L}_M\boldsymbol{D}_M\boldsymbol{U}_M$, where $\boldsymbol{L}_M$ (resp. $\boldsymbol{U}_M$) is a lower triangular (resp. upper triangular) matrix with 1 on the diagonal. The diagonal entries of $\Mat{G}_M$ are given
(with $K=1$) by
\begin{equation}\label{eq: RM Narula Cholesky}
d_{m}=1+P\left\vert a_{m}\right\vert ^{2}+P\left\vert
b_{m}\right\vert ^{2}\left( 1-P\frac{\left\vert a_{m-1}\right\vert
^{2} }{d_{m-1}}\right) \ ,\ m=2,\ldots ,M\ ,
\end{equation}
where the initial condition of \eqref{eq: RM Narula Cholesky} is
$d_{1}=1+P\left\vert a_{1}\right\vert ^{2}+P \left\vert
b_{1}\right\vert ^{2}$.
Thus, the  diagonal entries $\{d_{m}\}$ form
a discrete-time continuous space Markov chain; Narula's main observation
was that this chain possesses a unique
ergodic
stationary distribution, given by
\begin{equation}\label{eq: RM Narula
stationary distribution} f_{d}(x)=\frac{\log
(x)e^{-\frac{x}{P}}}{\text{Ei}\left( \frac{1}{\bar{ P}}\right)
P}\quad ;\quad x\geq 1\ ,
\end{equation}
where $\text{Ei}(x)=\int_{x}^{\infty }\frac{\exp (-t)}{t}dt$ is
the exponential integral function. Further, as is proved in
\cite{Narula-1997}, the strong
law of large numbers (SLLN) holds for the sequence $\{\log {
d_{m}}\}$ as $M\rightarrow \infty$.
Hence, the average per-cell sum-rate capacity of the
intra-cell TDMA scheme ($K=1$) can be expressed as
\begin{equation}\label{eq:
RM TDMA uplink capacity}
\begin{aligned}
R_{\mathrm{tdma-f}}&=\lim_{M\rightarrow\infty}\E\left(\frac{1}{M}\log\det
\Mat{G}_M\right)\\
&=\lim_{M\rightarrow\infty}\E\left(\frac{1}{M}\log\det\left(\boldsymbol{L}_M
\boldsymbol{D}_M\boldsymbol{U}_M\right)\right)\\
&=\lim_{M\rightarrow\infty}\E\left(\frac{1}{M}\sum_{m=0}^M \log
d_m\right) = \E_{\pi_d}\left(\log d\right)\ ,
\end{aligned}
\end{equation}
where the last expectation is taken with respect to $f_{d}(x)$, as
defined in \eqref{eq: RM Narula stationary distribution}. In particular,
\begin{equation}  \label{eq: RM TDMA uplink capacity explicit}
R_{\mathrm{tdma-f}}=\int_1^\infty\frac{(\log (x))^2
e^{-\frac{x}{P}} }{\text{Ei}\left( \frac{1}{P}\right) P}dx\ .
\end{equation}
Narula's approach is based on an explicit calculation of
the invariant distribution $f_d$, and is thus
tied to Rayleigh fading.
Modifications of key parameters
(such as the entries' PDF, and the number of nonzero
diagonals) lead to analytically intractable expressions.

Another result derived by following the footsteps of
\cite{Narula-1997} is an upper bound on the per-cell sum-rate of
the WB scheme with finite $K$ and infinite number of cells
$M\rightarrow\infty$, in the presence of a general fading
distribution, given by
\begin{equation}\label{eq: WB finite K ub}
 R_{\mathrm{wbk-f}}\leq \log \left(
\frac{1+2Pm_2+\sqrt{1+4Pm_2+4P^2\left(1-\frac{1}{K}\right)\left(m_2^2-\abs{m_1}^4\right)}}{2}\right)
\ .
\end{equation}
and in the special case of zero mean unit power ($m_1=0,\ m_2=1$)
fading distribution (e.g. Rayleigh fading) the bound reduces to
\begin{equation}\label{eq: WB finite K ub Rayligh fading}
R_{\mathrm{wbk-f}}\leq \log \left(
\frac{1+2P+\sqrt{(1+2P)^{2}-(4P^{2}/K)}}{2}\right) \ .
\end{equation}
This result which is proved in \cite{Narula-1997} for $K=1$
(intra-cell TDMA protocol) and expanded to an arbitrary $K$ in
\cite{Lifang-Goldsmith-Globecom06}, is derived by noting that the
average of the determinant of the received vector covariance
matrix $\Mat{G}_{M}$ can be recursively expressed by
\begin{equation}\label{eq: Narula diff eq}
    \E(\det \Mat{G}_m)  = A\ \E(\det \Mat{G}_{m-1})-B\ \E(\det\Mat{G}_{m-2})\
    ;\ m=3,\ldots,M\ ,
\end{equation}
with initial conditions
\begin{equation}\label{eq: Narula diff eq initial conditions}
\E(\det \Mat{G}_{1}) = A\quad ; \quad \E(\det \Mat{G}_{2})= A^2-B\
,
\end{equation}
where
\begin{equation}\label{eq: Narula diff eq consts}
A = 1+2P m_2\quad ; \quad B =
\frac{P^{2}}{K}\left(m_2^2+(K-1)\abs{m_1}^4\right)\ .
\end{equation}
See Appendix \ref{app: Jacobi det} for more
details. The solution to \eqref{eq: Narula diff eq} is given by
\begin{equation}\label{eq: General diff eq solution1}
    \E(\det \Mat{G}_m) = \varphi\ r^m - \phi\ s^m\ ,
\end{equation}
where
\begin{equation}\label{}
r = \frac{1}{2}\left(A+\sqrt{A^2-4B}\right)\quad ;\quad s =
\frac{1}{2}\left(A-\sqrt{A^2-4B}\right)\ ,
\end{equation}
are real and positive, and $\varphi$, $\phi$ are determined by the
initial conditions \eqref{eq: Narula diff eq initial conditions}.
Finally,
\eqref{eq: WB finite K ub} is derived by the following set of
inequalities
\begin{equation}
R_{\mathrm{wbk-f}}=\lim_{M\rightarrow \infty }\frac{1}{M}\E\left(
\log \det \boldsymbol{G}_{M}\right) \leq \lim_{M\rightarrow \infty
}\frac{1}{M} \log \E\left( \det \boldsymbol{G}_{M}\right) \ =\log
r\ ,
\end{equation}
where the inequality is due to Jensen's inequality, and the
last equality follows from the fact that $r>s$, and $M\rightarrow
\infty $. In the case of $K=1$, the upper bound of \eqref{eq: WB
finite K ub Rayligh fading} coincides with the per-cell sum-rate
capacity of the non-fading setup \eqref{eq: Rate no fading}.
Thus,
 the presence of Rayleigh fading
decreases the rates of the intra-cell TDMA protocol supported by
the ``soft-handoff" model. Nevertheless, it is shown in
\cite{Somekh-Zaidel-Shamai-IT-2005} that already for $K=2$ the
presence of fading may be beneficial at least for low SNR values.
The tightness of the bound is demonstrated by noting the for
$K\rightarrow \infty$ it coincides with the asymptotic
expression of \eqref{eq: WB large K rate}.

The extreme SNR characterization of the WB rate for
$M\rightarrow\infty$ in the presence of a general fading
distribution is summarized by \cite{Somekh-Zaidel-Shamai-IT-2005}
\begin{equation}\label{eq: WB extreme SNR general fading}
\begin{array}{cc}
S_{0}=\frac{2}{\frac{\mathcal{K}}{2K}+\frac{\abs{m_1}^4}{2
m_2^2}+1}\ ; &
\frac{E_{b}}{N_{0}}_{\min }= \frac{\log 2}{2 m_2}\\
S_{\infty }\le 1\ ; &
-\log_2\left(m_2+\sqrt{\left(1-\frac{1}{K}\right)\left(m_2^2-\abs{m_1}^4\right)}\right)\le\mathcal{L}_{\infty
}\ .
\end{array}
\end{equation}
The bounds of the high-SNR parameters are tight
for $K\gg 1$. For the special case of Rayleigh fading the extreme
SNR characterization are given by
\cite{Somekh-Zaidel-Shamai-IT-2005}
\begin{equation}  \label{eq: WB extreme SNR Rayleigh fading}
\begin{array}{cc}
S_{0}=\frac{2}{1+\frac{1}{K}}\ ; &
\frac{E_{b}}{N_{0}}_{\min }= \frac{\log 2}{2}\\
S_{\infty }=1\ ; & -\log_2\left(1+\sqrt{1-\frac{1
}{K}}\right)\le\mathcal{L}_{\infty }\le \frac{\gamma}{\log 2} \ ,
\end{array}
\end{equation}
where $\gamma\approx0.5772$ is the Euler-Mascheroni constant. It
is noted that the right inequality of the high-SNR power offset is
tight for $K=1$, while the left inequality is tight for $K\gg 1$.
The beneficial effects of Rayleigh fading and increasing number of
users are evident when compared to the non-fading extreme-SNR
parameters of the respective non-fading setup \eqref{eq: DL SH no
fading extreme SNR}.

To conclude this section we emphasize that calculating exact
expressions for the high-SNR parameters of the WB protocol rate
with finite number of users per-cell and \emph{general} fading
distribution remains an open problem.

\section{Applications}\label{sec: Applications}
In this section we present several applications of the main
results presented in this work (see Section \ref{sec: Main
results}).

\paragraph{Intra-Cell TDMA and Rayleigh Fading}

Assuming that only one user is active per-cell $K=1$ and symmetric
Rayleigh fading channels (i.e. $\pi_{\abs{a}^2}$ and
$\pi_{\abs{b}^2}$ are exponential distributions with parameter 1),
the high-SNR power offset is given according to Theorem
\ref{principaltext}, by
\begin{equation}\label{eq: Linf Natan}
    \Linf =
    -\max\left(\E(\log_2\abs{a}^2),\E(\log_2\abs{b}^2)\right)=\frac{-1}{\log
    2}\int_0^\infty e^{-x}\log x\  dx=\frac{\gamma}{\log 2}
\end{equation}
where the last equality is due to  \cite[pp.
567, formula 4.331.1]{Gradshteyn-Ryzhik-6th}. Obviously this result coincides with
the high-SNR power-offset derived by applying the definition of
$\Linf$ (see \eqref{eq: Linf Definition}) directly to the exact
expression derived in \cite{Narula-1997} (see expression
\eqref{eq: RM TDMA uplink capacity explicit}).

Note that the same result holds if an attenuation factor is added
to one of the fading paths, e.g. $\tilde{b}_m = \alpha b_m$ where
$b_m\sim \mathcal{CN}(0,1)$ and $\alpha\in[0,1]$; this follows
directly from Theorem \ref{principaltext}, but not from
\cite{Narula-1997},  which requires symmetric fading paths (i.e.
$\alpha=1$).

\paragraph{Intra-Cell TDMA and General Fading Statistic}
Consider the following single user single-input single-output
(SISO) flat fading channel for an arbitrary time index
\begin{equation}\label{eq: SISO channel}
    y=ax+n\ ,
\end{equation}
where $x$ is the input signal $x\sim \mathcal{CN}(0,P)$, and $n$
is the additive circularly symmetric Gaussian noise $n\sim
\mathcal{CN}(0,1)$. In addition, $a$ is the fading coefficient
$a\sim \pi_a$ satisfying conditions (H1)\dots(H3)
and known only  to the
receiver (receiver CSI). Assuming that the fading process is also
ergodic in the time domain, the ergodic capacity of the channel is
given by \cite{Biglieri-Proakis-Shamai-IEEE-T-IT-98}
\begin{equation}\label{eq: SISO channel capacity}
    C=\E_{\pi_a}\log(1+P\abs{a}^2)\ ,
\end{equation}
where the expectation is taken over the fading distribution
$\pi_a$. Accordingly, under the mild conditions (H1)\dots(H3),
the high-SNR
regime of this channel is characterized by
\begin{equation}\label{eq: SISO hi-snr}
    \Sinf = 1\quad ;\quad \Linf = -\E_{\pi_a}\log_2\abs{a}^2\ .
\end{equation}
Using Theorem \ref{principaltext},
we can now establish the following analogy
between the multi-cell setup and the SISO channel at hand.
\begin{cor}
The high-SNR characterization of the intra-cell TDMA per-cell
sum-rate supported by the ``soft-handoff" setup with fading
distributions $\pi_a, \pi_b$ such that
$\E_{\pi_a}\log_2\abs{a}^2>\E_{\pi_b}\log_2\abs{b}^2$, coincides
with those of a scalar single-user fading channel with fading
distribution $\pi_a$.
\end{cor}
This observation allows us to use the vast body of work done for
the celebrated scalar flat fading channel
\cite{Biglieri-Proakis-Shamai-IEEE-T-IT-98}. In particular, the
high-SNR characterization of flat fading channels with the
following fading statistics have been considered in previous
works: (a) Rayleigh distribution, (b) Rice distribution, (c)
log-normal distribution, and (d) Nakagami distribution (see
\cite{Biglieri-Proakis-Shamai-IEEE-T-IT-98} and references
therein).

\paragraph{Intra-Cell TDMA and Opportunistic Scheduling}

Throughout this work we have assumed that the instantaneous
channel state information is known to the MCP receiver only. Here
we further assume that some sort of ideal feedback channel is
available between the MCP receiver and the $K$ mobile users
included in each cell. This feedback channel is used to schedule
the ``best" local user in each cell for transmission during the
current time slot\footnote{See
\cite{Somekh-Simeone-Barness-Haimovich-Shamai-IT07} for a similar
scheduling deployed in the Wyner cellular uplink channel.}. In
other words, in each cell the user with the strongest channel fade
towards the BS located on the right boundary of each cell is
scheduled for transmission\footnote{Since the right most cell
indexed (M+1), has no BS on its right boundary it randomly
schedules a user for transmission.} with power $P$. Hence, the
index of the selected user in the $m$th cell reads
\begin{equation}\label{eq: Best user}
\tilde{k}_m=\underset{k=1,2,\ldots
K}{\mathrm{argmax}}\abs{a_{m,k}}^2\quad m=1,2,\ldots, M\ .
\end{equation}
The resulting $M\times (M+1)$ channel transfer matrix
$\Mat{\tilde{H}}_M$ of this scheduling scheme is a two diagonal
matrix with independent entries. The probability density function
of the main diagonal i.i.d. entries' amplitudes is given by
\begin{equation}\label{eq: Max PDF}
    d\pi_{K,\abs{a}^2}=K\pi_{\abs{a}^2}^{K-1}d\pi_{\abs{a}^2}\ ,
\end{equation}
following the maximum order statistics \cite{David-Book-1981}. On
the other hand, the i.i.d. entries of the second non-zero diagonal
are distributed according to the original fading statistics
$\pi_b$.

Assuming that $\pi_{K,\abs{a}^2}$ and $\pi_b$ satisfy conditions
(H1)\dots(H3), we can apply Theorem \ref{principaltext} in order
to derive the high-SNR characteristics of the per-cell sum-rate
achievable by this opportunistic scheduling
\begin{equation}\label{eq: Best local fade hi-snr}
    \Sinf=1\quad ; \quad
    \Linf=-\max\left(\E_{\pi_{K,\abs{a}^2}}(\log_2
    y),\E_{\pi_b}(\log_2 \abs{b}^2)\right)\ .
\end{equation}
For Rayleigh fading channels and in the case where the number of
users per-cell is large $K\gg 1$, we can use the well known fact
that the square of the maximum of the $K$ amplitudes behaves like
$\log K$ with high-probability (see
\cite{Sharif-Hassibi-Feb-2005}). Hence, the rate high-SNR power
offset of this scheme is
\begin{equation}\label{}
    \Linf \approx -\log_2\log K\ ,
\end{equation}
revealing a multi-user diversity gain of $\log\log K$. It is noted
that allowing additional power control to this scheme will yield
better performances. However, we are unable to apply Theorem
\ref{principaltext} for this situation. Finally, choosing the BS
located on the right boundary of the cell is arbitrary; taken the
BS located on the left boundary of the cell yields the same
results.

\section{Concluding Remarks}\label{sec: Concluding Remarks}

In this paper we study the high-SNR characterization of the
per-cell sum-rate capacity of the ``soft-handoff" uplink cellular
channel with multi-cell processing. Taking advantage of the special topology induced
by the setup, the problem reduces to the study of the spectrum of
certain large random Hermitian Jacobi matrices. For the intra-cell
TDMA protocol where only one user is active simultaneously
per-cell we provide an exact closed form expression for the
per-cell sum-rate high-SNR power offset for rather general fading
distribution. Examining the result, it is concluded that in the
high-SNR regime, the rate of the cellular setup at hand is
equivalent to the one of a single user SISO channel with similar
fading statistics.

Turning to the capacity achieving WB protocol, where all $K$ users
are active simultaneously in each cell, we derive a series of
lower and upper bounds to the rate. These bounds are shown (via
Monte-Carlo simulations) to be tighter than previously known bounds.

Note that in Theorem \ref{principalKtext} points a) and c) and in Proposition \ref{boundsK},
we take the fading coefficients relative to the users of one cell to be independent. Those results
continue to be true
if we assume correlation between the fading coefficients relative to the users of the same cell
(but independence between cells).
The proof is identical to the proof given in the paper.

Some of the analysis reported here can be extended to include the
case where $\Mat{G}_M$ is $(2p-1)$-diagonal for some $p>2$ (e.g.
$p=3$ for the channel matrix of the Wyner model), using an
adaptation of the ``Thouless formula for the strip" derived originally
in \cite{Craig-Simon}. Using this approach, bounds similar
to those of Prop. \ref{boundsK} may be provided on the rate.
Details will appear elsewhere \cite{nathanthesis}.

\section*{Acknowledgments}
N. L. was partially supported
by the fund for promotion of research at the Technion.

O. S. was partially supported
    by a Marie Curie Outgoing International Fellowship within the 6th
    European Community Framework Program.

S. S. was partially supported by the REMON Consortium and NEWCOM++.

O. Z. was partially supported by NSF grant number DMS-0503775;
Part of this work was done while he  was with the
Department of Electrical Engineering, Technion.

We thank David Bitton for his help with the implementation of the Monte Carlo
simulation used in Section \ref{sec: Numerical results}
(application of Proposition \ref{boundsK}).

\appendix

\subsection{Proof of Theorem \ref{principaltext}}
\label{sec:DefinitionsAndMainResult}

In order to streamline the proof we  somewhat modify notation. We
consider two random sequences of complex numbers $(a_n)$ and
$(b_n)$. The $(a_n)$ (resp. $(b_n)$) are i.i.d of law $\pi_a$
(resp. $\pi_b$) and the $(a_n)$ are independent of the $(b_n)$. We
set $\Omega\triangleq((a_n),(b_n))$. We denote by $\P$ the
probability associated with those random sequences and by $\E$ the
associated expectation. For a given integer $n$, we consider a
channel transfer matrix $\Mat{H}_M$ of size $M\times (M+1)$.

\[\Mat{H}_M=
\begin{pmatrix}
a_1&b_1&0&\cdots&0\\
0&\ddots&\ddots&\ddots&\vdots\\
\vdots&\ddots&\ddots&\ddots&0\\
0&\cdots&0&a_M&b_M\\
\end{pmatrix}.\]

We consider the following variable
\[\mathcal{C}_{M}(P)=\frac{1}{M}\tr\left\{\log\left(I+P \Mat{H}_M\Mat{H}_M^\dagger\right)\right\}.\]
Note that,
\[\Mat{H}_M\Mat{H}_M^\dagger=\begin{pmatrix}
\abs{a_1}^2+\abs{b_1}^2&a_2^\dagger b_1&0&\cdots&0\\
a_2b_1^\dagger &\abs{a_2}^2+\abs{b_2}^2&a_3^\dagger b_2&\ddots&\vdots\\
0&\ddots&\ddots&\ddots&0\\
\vdots&\ddots&\ddots&\ddots&a_M^\dagger b_{M-1}\\
0&\cdots&0&a_M b_{M-1}^\dagger &\abs{a_M}^2+\abs{b_M}^2\\
\end{pmatrix}.\]
With this notation, as explained in Section
\ref{sec: Main results},
Theorem \ref{principaltext} follows from the
following.
\begin{thm}
\label{principal}[$K=1$]
 Assume (H1) and (H2)\,.
\begin{enumerate}
    \item For every $\rho>0$, $\mathcal{C}_M(P)$ converges $\P$-a.s as
        $M$ goes to infinity. We call the limit $\mathcal{C}(P)$.
    \item Further assume [(H3) or
        (H3')]. As $\rho$ goes to infinity,
        \[\mathcal{C}(P)=\log P+
        2\max\left(\E_{\pi_a}\log\abs{x};\E_{\pi_b}\log\abs{x}\right)+o(1).\]
\end{enumerate}
\end{thm}
\noindent {\it Proof of Theorem \ref{principal}} Without loss of
generality, in the proof we can assume
\begin{enumerate}
  \item[(H5)] $\E_{\pi_a}\log\abs{x}\leq\E_{\pi_b}\log\abs{x}$.
\end{enumerate}
Indeed, we may exchange the role of entries $a_i$ and $b_i$ for $1\leq
i\leq M$ by a right-left reflection, namely the transformation
$\hat a_j=b_{M-j+1}$, $\hat b_j=a_{M-j+1}$, $1\leq j\leq M$.

For part a), only  (H1) and (H2) are needed. Since part a) is a
consequence of general facts concerning products of random
matrices and does not use much of the special structure in the
problem, we bring it in Appendix
\ref{sec:ProductOfRandomMatrices}.

Part b) uses the theory of Markov chains and is specific to the
particular matrix $\Mat{H}_M$. We note that as a by product of
this approach, we obtain a second proof of part a), however under
the additional assumption [(H3) or (H3')]. We provide a proof of
Theorem \ref{principal} under the assumptions (H1), (H2) and [(H3)
or (H3')] in Appendices \ref{sec:DefinitionsAndMainResult} and
\ref{sec:StudyOfTheMarkovChain}.

The structure of the proof is as follows.
We first introduce an auxiliary sequence which allows us to
reformulate the problem in terms of a special Markov chain.
The study of the latter, which forms the bulk of the proof of
Theorem \ref{principal}, is carried out in Section
\ref{sec:StudyOfTheMarkovChain}.

\subsubsection{Auxiliary sequence}
\label{sec:AuxiliarySequence} We begin with a technical lemma.

\begin{lem}
\label{discriminant} Assume (H2). $\P$-a.s,
$\Mat{H}_M\Mat{H}_M^\dagger$ does not have multiple eigenvalues.
\end{lem}
\begin{proof}
We let $D$ denote the discriminant of
$\Mat{H}_M\Mat{H}_M^\dagger$, it is a
polynomial in\\
$\{\abs{a_i}^2+\abs{b_i}^2,a_{i+1}b_i^\dagger,a_{i+1}^\dagger
b_i\}$ which vanishes when there is a multiple eigenvalue.
Therefore, it is a polynomial in $\Re a_i$, $\Im a_i$, $\Re b_i$
and $\Im b_i$ It is not identically 0 because for $b_i=0$ and
$a_i=i$, the eigenvalues of $\Mat{H}_M\Mat{H}_M^\dagger$ are
distinct. The result follows directly from the following lemma
which is an easy consequence of Fubini's theorem.
\end{proof}
\begin{lem}
\label{giab} Let $Q$ be a function from $\C^n$ to $\C$. We assume
that $Q$ is not identically 0 and that $Q(z_1,\ldots,z_n)$ is a
polynomial in the $\Re z_i$ and the $\Im z_i$. Then the set of the
roots of $Q$ has Lebesgue measure 0.
\end{lem}

In the sequel, we denote by $\lambda_1\geq \ldots\geq \lambda_M$ the
ordered eigenvalues of $\Mat{H}_M\Mat{H}_M^\dagger$. For a given
$\lambda$, we consider the following sequence (indexed by $n$) of
complex numbers (the dependence in $\lambda$ will only be
mentioned when it is relevant): $x_0=0$, $x_1=1$, and for
$n\geq1$,
\[a_n b_{n-1}^\dagger x_{n-1}+(\abs{a_n}^2+\abs{b_n}^2)x_n+a_{n+1}^\dagger b_{n}x_{n+1}=\lambda x_n,\]
that is
\begin{align}
\label{xn1}
x_{n+1}=\frac{\lambda-\abs{a_n}^2-\abs{b_n}^2}{a_{n+1}^\dagger
b_{n}} x_n-\frac{a_n b_{n-1}^\dagger}{a_{n+1}^\dagger
b_{n}}x_{n-1}.
\end{align}
Note that $x_{M+1}(\lambda)=0$ if and only if $\lambda$ is an
eigenvalue of $\Mat{H}_M\Mat{H}_M^\dagger$. Moreover, $x_{n+1}$ is
a polynomial in $\lambda$ of degree $n$ with highest coefficient
$1/\prod_{i=1}^{n}(a_{i+1}^\dagger {b_i})$. One can thus write
using Lemma \ref{discriminant}
\[x_{n+1}(\lambda)=\prod_{i=1}^{n}(a_{i+1}^\dagger{b_i})^{-1}\prod_{i=1}^{n}(\lambda-\lambda_i)\ \ \ \ \P-\textrm{a.s},\]

Hence, for $\lambda=-1/P$,
\begin{align}
\label{xn}
\mathcal{C}_M(P)=\log(P)+\frac{1}{M}\log\abs{x_{M+1}(\lambda)}+\frac{1}{M}\sum_{i=1}^M\log\abs{a_{i+1}b_i}\
\ \ \ \ \ \ \P-\textrm{a.s}.
\end{align}
By the Law of Large Numbers (LLN),
\[\lim_{M\rightarrow\infty}\frac{1}{M}\sum_{i=1}^M\log\abs{a_{i+1}b_i}=\E_{\pi_a}\log\abs{x}+\E_{\pi_b}\log\abs{x}\ \ \ \ \P-\textrm{a.s}.\]

Because of (\ref{xn}), to prove Theorem \ref{principal}, we only
need to show the following lemma.
\begin{lem}
\label{gros} Assume  (H1), (H2) and [(H3) or (H3')]
\begin{enumerate}
    \item \label{gr1} For every $\lambda<0$, $\frac{1}{n}\log\abs{x_{n+1}(\lambda)}$ converges $\P$-a.s as $n$ goes to infinity. The limit is $\gamma(\lambda)$, the Lyapunov exponent defined by (\ref{fur}).
    \item \label{gr2} Assume further (H5).
        Then $\gamma(\lambda)$ converges to $\E_{\pi_b}\log\abs{x}-\E_{\pi_a}\log\abs{x}$ as
        $\lambda$ goes to 0.
\end{enumerate}
\end{lem}

\subsubsection{Reduction to a Markov chain}
\label{sec:ReductionToAMarkovChain}

To prove Lemma \ref{gros}, we take $c_n\triangleq x_n/x_{n-1}$,
for $n\geq2$. Note that by (\ref{xn1}) and (H2), $\P$-a.s,
$x_n\neq0$, hence $c_n$ is well defined and non-zero. By
(\ref{xn1}), we get
\[c_{n+1}=\frac{\lambda-\abs{a_n}^2-\abs{b_n}^2}{a_{n+1}^\dagger b_{n}}-\frac{a_n b_{n-1}^\dagger}{c_n a_{n+1}^\dagger b_{n}}.\]
Let $d_n=c_n a_n^\dagger b_{n-1}$. Then,
\[d_{n+1}=\lambda-\abs{a_n}^2-\abs{b_n}^2-\frac{\abs{a_n}^2\abs{b_{n-1}}^2}{d_n}=\lambda-\abs{b_n}^2-\abs{a_n}^2\left(1+\frac{\abs{b_{n-1}}^2}{d_n}\right).\]
Let $e_{n}=\left(1+\frac{\abs{b_{n-1}}^2}{d_n}\right)$. Then
$d_{n+1}=\lambda-\abs{b_n}^2-\abs{a_n}^2e_{n}$, and
\begin{align}
\label{en}
e_{n}=\frac{-\lambda+\abs{a_{n-1}}^2e_{n-1}}{-\lambda+\abs{b_{n-1}}^2+\abs{a_{n-1}}^2e_{n-1}},
\end{align}
with the initial conditions,
\[c_2=\frac{\lambda-\abs{a_1}^2-\abs{b_1}^2}{a_2^\dagger b_1};\]
\[d_2=\lambda-\abs{b_1}^2-\abs{a_1}^2.\]
$d_2\in\R$ and $d_2<-\abs{b_1}^2$, hence, $0<e_2<1$. From
(\ref{en}) we conclude that for all $n$, $e_n\in\R$ and $0<e_n<1$.
Now, for all $n$,
\[c_n=\frac{d_n}{a_n^\dagger b_{n-1}}=\frac{b_{n-1}^\dagger}{a_n}^\dagger\frac{1}{e_n-1}.\]
Then,
\begin{equation}
\label{goal}
\begin{split}
\frac{1}{n}\log\abs{x_{n+1}}&=\frac{1}{n}\sum_{i=2}^{n+1}\log\abs{c_i}\\
&=\frac{1}{n}\sum_{i=2}^{n+1}\left(\log\abs{\frac{b_{i-1}}{{a_i}}}-\log(1-e_i)\right)\\
\end{split}
\end{equation}
$\frac{1}{n}\sum_{i=2}^{n+1}\log\abs{\frac{b_{i-1}}{{a_i}}}$
converges to $\E_{\pi_b}\log\abs{x}-\E_{\pi_a}\log\abs{x}$ by the
LLN. We now study in details the Markov chain $e_n$.

\subsection{Study of the Markov chain $e_n$ and proof
of Lemma \ref{gros}} \label{sec:StudyOfTheMarkovChain}

For simplicity, we write $\delta\triangleq-\lambda$ and we
re-index the chain so that it starts from $e_0$. As in (\ref{en}),
\begin{align}
\label{en1}
e_{n}=\frac{\delta+\abs{a_{n-1}}^2e_{n-1}}{\delta+\abs{b_{n-1}}^2+\abs{a_{n-1}}^2e_{n-1}}.
\end{align}
We denote by $\P_{e_0}$ the law of the sequence starting from
$e_0$ and by $\E_{e_0}$ the associated expectation.

\begin{prop}
\label{ergo} Assume (H2) and [(H3) or (H3')]. The Markov chain
$e_n$ has a unique stationary probability, say, $\mu_\delta$ and
for $s\in \L^1(\mu_\delta)$, for every starting point
$e_0\in[0,1]$, $\P_{e_0}$-a.s,
\[\frac{1}{n}\sum_{i=0}^{n}s(e_i)\tendvers\int s d\mu_\delta.\]
\end{prop}

\begin{proof}
We start with two lemmas that will be proved later on.

\begin{lem}
\label{hx3} For $\alpha,\beta,\delta\in\R^+$, we define the
function $\phi_{\alpha,\beta}$ (we suppress $\delta$ from the
notation) such that for $e\in[0,1]$
\[\phi_{\alpha,\beta}(e)=\frac{\delta+\alpha e}{\delta+\beta+\alpha e}.\]
For any given $e\in[0,1]$, we define the sequence $(\theta_n(e))$
by $\theta_0=e$ and for $n\geq 1$,
$\theta_n(e)=\phi_{\alpha,\beta}(\theta_{n-1}(e))$. Then,
$\phi_{\alpha,\beta}$ has exactly one fixed point in $[0,1]$, say
$\kappa_{\alpha,\beta}$, and $\theta_n(e)$ converges to
$\kappa_{\alpha,\beta}$. Moreover, the convergence is uniform in
the starting point in the following sense:
\[(\forall \ep>0)(\exists n_0\in\N)(\forall e\in[0,1])(\forall n\geq n_0)(\abs{\theta_n(e)-\kappa_{\alpha,\beta}}<\ep).\]
Finally if $\alpha_1<\alpha_2$ and $\beta_1>\beta_2$, then
$\kappa_{\alpha_1,\beta_1}<\kappa_{\alpha_2,\beta_2}$.
\end{lem}

\begin{lem}
\label{bornes} Assume (H2) and [(H3) or (H3')].
\begin{enumerate}
    \item For $e_0\in[0,1]$, there exist two sequences $(\theta_n^1(e_0))$ and $(\theta_n^2(e_0))$ in $[0,1]$ such that the law of $e_n$ under $\P_{e_0}$ and the Lebesgue-measure on $[(\theta_n^1(e_0)),(\theta_n^2(e_0))]$ are mutually absolutely continuous.
    \item $(\theta_n^1(e_0))$ and $(\theta_n^2(e_0))$ converge to, say $\Theta^1$ and $\Theta^2$ respectively, $\Theta^1$ and $\Theta^2$ are independent of $e_0$ and $\Theta^1<\Theta^2$. Finally, the convergence is uniform in the starting point in the sense of Lemma \ref{hx3}.
  \item If $e_0\in[\Theta^1,\Theta^2]$, then for all $n$, the law of $e_n$ under $\P_{e_0}$ is absolutely continuous with respect to the Lebesgue-measure on $[\Theta^1,\Theta^2]$.
\end{enumerate}
\end{lem}

We recall some definitions from the
theory of Harris Markov chains, which will be used extensively in the proof.
We refer the reader to
\cite{harris} for the relevant background.

\begin{deft}
\label{irreduc} Denote by $(r_n)$ a Markov chain on $I$ an
interval of $\R$. Set $l$ a probability measure on $I$, it is an
\emph{irreducibility measure} if for all measurable set $A$ such
that $l(A)>0$ and for all $r_0\in I$
\[(\exists n)\ \P_{r_0}(r_n\in A)>0.\]
$l$ is a \emph{maximal irreducibility measure} if it satisfies the
following conditions:
\begin{itemize}
    \item $l$ is an irreducibility measure.
    \item For any other irreducibility measure $l'$, $l'$ is absolutely continuous with respect to $l$.
    \item If $l(A)=0$ then $l\{r_0: (\exists n)\ \P_{r_0}(r_n\in A)>0\}=0$.
    \item For any irreducibility measure $l'$, $l$ is equivalent to
       \[\int_I l'(dr_0)\sum_{i=0}^\infty\frac{1}{2^i}\P_{r_0}(r_i\in\cdot).\]
\end{itemize}
\end{deft}

\begin{deft}
Denote by $(r_n)$ a Markov chain on $I$ an interval of $\R$. A set
$A$ is called \emph{Harris recurrent} if for all $r_0\in A$,
$\P_{r_0}$-a.s, the chain $r_n$ visits $A$ an infinite number of
times. The chain $(r_n)$ is called \emph{Harris recurrent} if
given a maximal irreducibility measure $l$, every measurable set
$A$ such that $l(A)>0$ is Harris recurrent.
\end{deft}

\begin{deft}
Denote by $(r_n)$ a Markov chain on $I$ an interval of $\R$.
Denote by $l$ a maximal irreducibility measure. For every
measurable set $A$ such that $l(A)>0$ we denote by $\tau_A$ the
time when the chain $(r_n)$ enters $A$. A measurable set $B$ is
called \emph{regular} if for every measurable set $A$ such that
$l(A)>0$,
\[\sup_{r_0\in B}\E_{r_0}(\tau_A)<\infty.\]
\end{deft}

\begin{deft}
\label{regular} Denote by $(r_n)$ a Markov chain on $I$ an
interval of $\R$. Denote by $A$ and $B$ two measurable sets. We
say that $B$ is \emph{uniformly accessible} from $A$ if there
exists an $\ep>0$ such that
\[\inf_{r_0\in A}\P_{r_0}((\exists n)\ r_n\in B)\geq\ep.\]
\end{deft}

We continue with the proof of Proposition \ref{ergo}. Denote by
$l$ the Lebesgue-measure on $[\Theta^1,\Theta^2]$. By
\cite[Theorem 17.0.1]{harris}, it is enough to prove that the
Markov chain $e_n$ is $l$-irreducible, positive Harris with
invariant probability $\mu_\delta$. Denote ${\cal B}^+$ the set of
Lebesgue-measurable subsets of $[0,1]$ with positive $l$-measure.
Here is a technical lemma that will be proved later on.

\begin{lem}
\label{clef} Assume (H2) and [(H3) or (H3')]. For all $B\in{\cal
B}^+$, there exists $n_0=n_0(B)$ such that for all $n\geq n_0$,
\[p_n\triangleq\inf_{e_{0}\in[0,1]}\P_{e_0}(e_{n}\in B)>0.\]
\end{lem}

We continue with the proof of Proposition \ref{ergo}.

\emph{Step 1: The Markov chain $e_n$ is $l$-irreducible, Harris
and admits an invariant measure unique up to a constant multiple.}
By Lemma \ref{clef}, for $e_0\in[0,1]$ and $B\in{\cal B}^+$, the
chain has a positive probability to reach $B$ in $n_0$ steps
starting from $e_0$. Therefore, the Markov chain $e_n$ is
$l$-irreducible and by Lemma \ref{bornes} c), $l$ is a maximal
irreducibility measure for the chain $e_n$. For a given $B\in{\cal
B}^+$, by Lemma \ref{clef}, the chain $e_n$ has a probability at
least $p_{n_0}$ to reach $B$ in $n_0$ steps, hence the chain will
eventually reach $B$ and hence come back to $B$ an infinite number
of times, therefore $B$ is Harris-recurrent and the Markov chain
$e_n$ is Harris. By \cite[Theorem 10.0.1]{harris}, the Markov
chain $e_n$ admits an invariant measure unique up to a constant
multiple.

\vskip 7pt

\emph{Step 2: The Markov chain $e_n$ is aperiodic.} By
\cite[Theorem 5.4.4]{harris}, there exists an integer $d$, the
period of the chain, such that there exist disjoint measurable
sets $D_0,\ldots,D_{d-1}$ such that
\begin{itemize}
    \item For $i=0\ldots d-1$, if $e_i\in D_i$, then $\P_{e_i}(e_{i+1}\in D_{i+1})=1$ (mod $d$).
    \item $l\left((\cup_{i=1}^d D_i)^c\right)=0$.
\end{itemize}
By Lemma \ref{bornes}, for $n_1\geq n_0$ large enough and $n\geq n_1$, the
Lebesgue-measure on
$J\triangleq[(2\Theta^1+\Theta^2)/3,(\Theta^1+2\Theta^2)/3]$ is
absolutely continuous with respect to the law of $e_n$ under
$\P_{e_0}$. Therefore, for any
$n\geq n_1$, if $e_n\in D_i$, then $J\subset D_i$,
and then, if $d>1$, $e_{n+1}\in D_{i+1}$ and thus
also $J\subset D_{i+1}$, a contradiction.
Hence, $d=1$.

\vskip 7pt

\emph{Step 3: The set $[0,1]$ is regular for the Markov chain
$e_n$.} Take $B\in{\cal B}^+$. By Lemma \ref{clef}, the time it
will take for the chain $e_n$ to enter $B$ is a.s bounded above by
$n_0$ times a geometric random variable of parameter $p_{n_0}$,
hence it expectation is bounded above by $n_0/p_{n_0}$, hence
$[0,1]$ is regular.

\vskip 7pt

Now we apply \cite[Theorem 13.0.1]{harris} and get that the Markov
chain $e_n$ is positive Harris, hence has a unique invariant
probability that we denote $\mu_\delta$.

\end{proof}

\begin{proof}[Proof of Lemma \ref{clef}]

The Lebesgue-measure on $[\Theta^1,\Theta^2]$ is regular hence
there exists an $\ep>0$ such that
$B\cap[\Theta^1+\ep,\Theta^2-\ep]$ has positive Lebesgue-measure.
By Lemma \ref{bornes} a) and b), we can take $n_0$ such that for
any given $n\geq n_0$ and any given starting point $e_0$,
$\P_{e_0}(e_n\in B)>0$. Fix $n\geq n_0$. Set
$\psi(e_0)=\P_{e_0}(e_n\in B)$. By (H2), $\psi$ is a continuous
function on $[0,1]$. By compactness,
\[\inf_{e_0\in[0,1]}\P_{e_0}(e_n\in B)>0.\]
\end{proof}

\begin{proof}[Proof of Lemma \ref{bornes}]
Let us start assuming (H3').

a) We first assume that ${\cal M}_a,{\cal M}_b\in\R^+$. We use the
notation of Lemma \ref{hx3}. For $e_0\in[0,1]$ and $n$, we define
$\theta_n^1(e_0)=\phi_{m_a,{\cal M}_b}^n(e_0)$ and
$\theta_n^2(e_0)=\phi_{{\cal M}_a,m_b}^n(e_0)$, where $\phi^n$ is
the $n$-th iteration of the function $\phi$. Note that for $e1\leq
e_2\in[0,1]$, $\alpha 1<\alpha_2\in\R^+$ and $\beta
1<\beta_2\in\R^+$,
\begin{align*}
\psi:[e_1,e_2]\times[\alpha_1,\alpha_2]\times[\beta_1,\beta_2]&\longrightarrow[\phi_{\alpha_1,\beta_2}(e_1),\phi_{\alpha_2,\beta_1}(e_2)]\\
(x,\alpha,\beta)&\longmapsto\phi_{\alpha,\beta}(e)
\end{align*}
is well defined and onto and the inverse image of an interval
which is not a singleton has positive Lebesgue-measure. Therefore,
by induction, the Lebesgue-measure on $[\theta_n^1,\theta_n^2]$ is
absolutely continuous with respect to the law of $e_n$ under
$\P_{e_0}$. Moreover, by (H2) and (\ref{en1}), the
Lebesgue-measure on $[\theta_n^1,\theta_n^2]$ and the law of $e_n$
under $\P_{e_0}$ are mutually absolutely continuous.

b) It is a direct consequence of Lemma \ref{hx3} and we get
$\Theta^1=\kappa_{m_a,{\cal M}_b}$ and $\Theta^2=\kappa_{{\cal
M}_a,m_b}$. By Lemma \ref{hx3} and (H3'), $\kappa_{m_a,{\cal
M}_b}<\kappa_{{\cal M}_a,m_b}$, hence $\Theta^1<\Theta^2$.

c) $\phi_{m_a,{\cal M}_b}$ is increasing and $\kappa_{m_a,{\cal
M}_b}$ a fixed point hence if $\kappa_{m_a,{\cal M}_b}\leq e_0$,
then for all $n$, $\kappa_{m_a,{\cal M}_b}\leq \theta_n^1(e_0)$.
In the same way, for all $n$, $\kappa_{{\cal M}_a,m_b}\geq
\theta_n^2(e_0)$.

If ${\cal M}_a=\infty$ (resp. ${\cal M}_b=\infty$), we take for
all $n\geq1$, $\theta_n^2=1$ (resp. $\theta_n^1=0$) and
$\Theta^2=1$ (resp. $\Theta^1=0$) and the proof is the same.

Let us now assume (H3). The proof is the same with for all
$n\geq1$ and all $e_0\in[0,1]$, $\theta_n^1(e_0)=0$, for all
$n\geq1$ and all $e_0\in[0,1]$ (except for $n=1$ and $e_0=0$),
$\theta_n^2(e_0)=0$. We get $\Theta^1=0$ and $\Theta^2=1$.
\end{proof}

\begin{proof}[Proof of Lemma \ref{hx3}]
For $e\in[0,1]$,
\[\phi_{\alpha,\beta}'(e)=\frac{\alpha\beta}{(\delta+\beta+\alpha e)^2}.\]
$\phi_{\alpha,\beta}'$ is decreasing and
$\phi_{\alpha,\beta}'(1)<1$. If $\phi_{\alpha,\beta}'(0)<1$, then
$\phi_{\alpha,\beta}$ is contracting hence admits a fixed point
and its iteration on any starting point converges to the fixed
point. Suppose $\phi_{\alpha,\beta}'(0)\geq1$. Denote by
$\overline{e}$ the only point of $[0,1]$ such that
$\phi_{\alpha,\beta}'(\overline{e})=1$. Set
$\tilde{\phi}_{\alpha,\beta}(e)=\phi(e)_{\alpha,\beta}-e$. Then
$\tilde{\phi}_{\alpha,\beta}(0)>0$,
$\tilde{\phi}_{\alpha,\beta}(1)\leq0$, and
$\tilde{\phi}_{\alpha,\beta}$ is increasing on $[0,\overline{e}]$
and decreasing on $[\overline{e},1]$. Hence,
$\tilde{\phi}_{\alpha,\beta}(\overline{e})>0$ and
$\tilde{\phi}_{\alpha,\beta}$ is 0 on exactly one point which is a
fixed point for $\phi_{\alpha,\beta}$. We denote that fixed point
$\kappa_{\alpha,\beta}$. If $e\in[\kappa_{\alpha,\beta},1]$, since
$\phi_{\alpha,\beta}$ is increasing, for all $n$,
$\theta_n(e)\in[\kappa_{\alpha,\beta},1]$ and
$\phi_{\alpha,\beta}$ is contracting on
$[\kappa_{\alpha,\beta},1]$ hence $\theta_n(e)$ converges to
$\kappa_{\alpha,\beta}$. If $e\in[0,\kappa_{\alpha,\beta}]$, for
all $n$, $\theta_n(e)\in[0,\kappa_{\alpha,\beta}]$, and
$\tilde{\phi}_{\alpha,\beta}$ is non-negative on that interval,
hence $\theta_n(e)$ is non-decreasing. Therefore, it converges and
since $\phi_{\alpha,\beta}$ is continuous, the only possible limit
is $\kappa_{\alpha,\beta}$. To prove the uniformity in the
starting point, we use the fact that $\phi_{\alpha,\beta}$ is
increasing, hence for all $e\in[0,1]$ and $n$,
\[\theta_n(0)\leq\theta_n(e)\leq\theta_n(1).\]
That gives the uniformity. Finally, assume $\alpha_1<\alpha_2$ and
$\beta_1>\beta_2$. $\phi_{\alpha,\beta}(e)$ is non-decreasing in
$\alpha$, decreasing in $\beta$ and non-decreasing in $e$ hence by
induction,
$\phi_{\alpha_1,\beta_1}^n(0)\leq\phi_{\alpha_2,\beta_2}^n(0)$,
where $\phi^n$ is the $n$-th iteration of the function $\phi$.
Hence, $\kappa_{\alpha_1,\beta_1}\leq\kappa_{\alpha_2,\beta_2}$.
If $\kappa_{\alpha_1,\beta_1}=\kappa_{\alpha_2,\beta_2}$, then
\[\kappa_{\alpha_1,\beta_1}=\phi_{\alpha_1,\beta_1}(\kappa_{\alpha_1,\beta_1})<\phi_{\alpha_2,\beta_2}(\kappa_{\alpha_1,\beta_1})=\phi_{\alpha_2,\beta_2}(\kappa_{\alpha_2,\beta_2})=\kappa_{\alpha_2,\beta_2},\]
which gives a contradiction.

\end{proof}

We continue with the proof of Lemma \ref{gros}. Recall that $0\leq
e_n\leq1$, hence $\mu_\delta$ is stochastically dominated by an
atom at $1$. $\mu_\delta$ is the invariant measure, since the
function $\phi_{\alpha,\beta}(\cdot)$ is increasing in $e$,
$\mu_\delta$ is stochastically dominated by the law of the chain
started at $1$ after one step:
\[\mu_\delta\preceq{\cal L}\left(\frac{\delta+\abs{a_{0}}^2}{\delta+\abs{b_{0}}^2+\abs{a_{0}}^2}\right)\preceq    {\cal L}\left(\frac{\abs{a_{0}}^2}{\abs{b_{0}}^2+\abs{a_{0}}^2}\right).\]
Thus, denoting by $\pi_0$ the law of
$\frac{\abs{a_{0}}^2}{\abs{b_{0}}^2+\abs{a_{0}}^2}$, and using
(H1),
\[\int-\log(1-x)d\mu_\delta(x)\leq\int-\log(1-x)d\pi_0(x)<\infty.\]
That is
\begin{align}
\label{L1} -\log(1-\cdot)\in \L^1(\mu_\delta).
\end{align}
With Proposition \ref{ergo}, we get
\begin{align}
\label{A}
\frac{1}{n}\sum_{k=2}^{n+1}-\log(1-e_k)\tendvers\int_0^1-\log(1-x)d\mu_\delta(x)\
\ \ \ \ \ \P_{e_2}-\textrm{a.s}.
\end{align}
With (\ref{goal}), it gives a proof of Lemma \ref{gros} a).

Let us prove Lemma \ref{gros} b). Take $\eta>0$ and $\ep>0$ small.
\begin{equation}
\label{trois}
\begin{split}
&\int_0^1-\log(1-x)d\mu_\delta(x)\\
&=\int_0^{\ep}-\log(1-x)d\mu_\delta(x)+\int_{\ep}^{1-\eta}-\log(1-x)d\mu_\delta(x)+\int_{1-\eta}^1-\log(1-x)d\mu_\delta(x)\\
&\leq-\ep\log(1-\ep)-\log\eta \mu_\delta([\ep,1])+\int_{1-\eta}^1-\log(1-x)d\mu_\delta(x).\\
\end{split}\end{equation}
By (\ref{L1}), the last term converges to 0 as $\eta$ goes to 0.
By (\ref{goal}), (\ref{A}) and (\ref{trois}), to prove Lemma
\ref{gros} b), we only have to prove that for any given $\ep>0$,
\[\mu_\delta([\ep,1])\xrightarrow[\delta\rightarrow0]{}0.\]
For that, by Proposition \ref{ergo}, we need to show that the
proportion of the time that the chain $e_n$ spends above $\ep$
converges to 0 as $\delta$ goes to 0. We take $0<\ep<\ep_0<1$,
where $\ep_0$ will be chosen later. We consider the Markov chain
$z_n\triangleq\log e_n$ and the random function $g_n$ such that
$z_n=g_n(z_{n-1})$. It is enough to show that the proportion of
the time that $z_n$ spends above $\log\ep$ goes to $0$ as $\delta$
goes to $0$. Let us couple $z_n$ with another Markov chain $w_n$,
such that $w_n\geq z_n$ a.s. and that the proportion of the time
that $w_n$ spends above $\log\ep$ goes to $0$ as $\delta$ goes to
$0$.

For that, we need good information on the jumps of $z_n$.

\begin{lem}
\label{jump} Assume (H1) and (H5). Set
\begin{align*}
j_n(z_{n-1})&\triangleq z_n-z_{n-1}\\
&=\log\left(\frac{\delta}{e^{z_{n-1}}}+\abs{a_{n-1}}^2\right)-\log\left(\delta+\abs{b_{n-1}}^2+\abs{a_{n-1}}^2e^{z_{n-1}}\right).
\end{align*}
$\left(\forall\delta>0\right)\left(\exists\ep'>0\right)\left(\forall
x\geq\log\ep'\right)$
\begin{enumerate}
    \item $\E j_n(x)\leq0$,
    \item $\var j_n(x)\leq V\triangleq\E\left(\left(\log(\abs{a_{n-1}}^2+\abs{b_{n-1}}^2)\right)^2+\left(\log(\abs{a_{n-1}}^2)\right)^2\right)+C$.
\end{enumerate}
C is a constant independent of everything. $\ep'$ is a function of
$\delta$ but we will not write it to keep the notation clear.
Moreover,
\[\lim_{\delta\rightarrow0}\ep'=0.\]
\end{lem}

The proof will be done at the end of the section.

We continue with the proof of Lemma \ref{gros} b). We take
$\delta>0$ such that $0<\ep'<\ep<\ep_0<1$. We define $w_n$ in a
way that it stays between $\log \ep'$ and $0$. Set $w_0=z_0$, for
$\delta$ small enough, $w_0>\log{\ep'}$. For $x\in[\log\ep';0]$,
denote
\[h_n(x)=g_n(x)-\E j_n(x)\geq g_n(x).\]
That is
\begin{equation}
\label{hn}
\begin{split}
&h_n(x)=x+\log\left(\frac{\frac{\delta}{e^{x}}+\abs{a_{n-1}}^2}{\delta+\abs{b_{n-1}}^2+\abs{a_{n-1}}^2e^{x}}\right)-\\
&\ \ \ \ \ \ \ \ \ \ \ \ \ \ \ \ \ \ \ \
\E\log\left(\frac{\frac{\delta}{e^{x}}+\abs{a_{n-1}}^2}{\delta+\abs{b_{n-1}}^2+\abs{a_{n-1}}^2e^{x}}\right).
\end{split}
\end{equation}
Note that
\begin{align}
\label{mart} \E(h_n(z_{n-1})-z_{n-1}|z_{n-1})=0.
\end{align}
\begin{itemize}
    \item If $h_n(w_{n-1})>0$, set $w_{n}=0$.
    \item If $h_n(w_{n-1})<\log\ep'$, set $w_{n}=\log\ep'$.
    \item Otherwise, set $w_n=h_n(w_{n-1})$.
\end{itemize}

In the first two case, we say that the chain is \emph{truncated}.
Note that for all $n$, $w_n\geq z_n$. Indeed, either $w_n=0\geq
z_n$ or $w_n\geq h_n(w_{n-1})\geq g_n(w_{n-1})\geq
g_n(z_{n-1})=z_n$, by induction and using the fact that $g_n$ is
a.s non-decreasing. Therefore, the proportion of the time that the
chain $w_n$ spends above $\log\ep$ is larger that the proportion
of the time that chain $z_n$ spends above $\log\ep$.

\begin{prop}
\label{ergo2} Assume (H2).
\begin{enumerate}
    \item The Markov chain $w_n$ has a unique stationary probability, say, $\nu_\delta$ and for $s\in L^1(\nu_\delta)$, for every starting point $w_0\in[\log\ep',0]$, $\P_{w_0}$-a.s,
\[\frac{1}{n}\sum_{i=0}^{n}s(w_i)\tendvers\int s d\nu_\delta.\]
  \item We denote $T$ the return time to $0$, starting from $0$. Then $\nu_\delta(0)=1/\E_0 T$.
\end{enumerate}
\end{prop}

\begin{proof}
See \cite{harris} and Definitions \ref{irreduc}-\ref{regular} for
the theory of Harris Markov chains that we will use extensively in
the proof. Define the following probability measure on
$[\log\ep',0]$. For $B$ a Borel set,
\[\overline{l}(B)\triangleq\sum_{n=0}^{\infty}\frac{1}{2^{n+1}}\P_0(w_n\in B).\]
Let us prove that the Markov chain $w_n$ is
$\overline{l}$-irreducible, positive Harris with invariant
probability $\nu_\delta$. By \cite[Theorem 17.0.1]{harris}, that
will prove a). We use the following lemma that will be proved
later on.

\begin{lem}
\label{rec0} Assume (H2).
\begin{enumerate}
    \item There exist $c>0$ and $\theta>0$ such that for all $x\in[\log\ep';0]$,
      \[\P\left(h_n(x)\geq x+c\right)>\theta.\]
  \item Set $N=\left\lceil\frac{-\log\ep'}{c}\right\rceil$. 0 is a recurrent point for the chain $w_n$ and the time between two visits at 0 is a.s bounded above by $N$ times a geometric random variable of parameter $\theta^N$.
\end{enumerate}
\end{lem}

We continue with the proof of Lemma \ref{ergo2}. The sets which
have positive $\overline{l}$-measure are exactly the sets that
have a positive probability to be visited starting from 0.
Moreover 0 is a recurrent point. Therefore, the Markov chain $w_n$
is $\overline{l}$-irreducible and $\overline{l}$ is a maximal
irreducibility measure. Moreover, take $B$ with positive
$\overline{l}$-measure, $B$ is uniformly accessible from $\{0\}$.
Therefore, we can apply \cite[Theorem 9.1.3 (i)]{harris} and since
0 is Harris-recurrent, $B$ is also Harris-recurrent, therefore,
the chain $w_n$ is Harris-recurrent. By Lemma \ref{rec0} b), the
time between two visits at 0 has finite expectation (bounded above
by $N/\theta^n$). Therefore, by \cite[Theorem 10.2.2]{harris}, the
chain $w_n$ is positive-Harris and admits a unique invariant
probability measure. That finishes the proof of point a). The
point b) is a consequence of
\[1=\nu_\delta([\log\ep',0])=\nu_\delta(0)\E_0[T],\]
which comes from \cite[Theorem 10.0.1]{harris}, which we apply to
$A=\{0\}$, which has positive $\overline{l}$-measure.
\end{proof}

\begin{proof}[Proof of Lemma \ref{rec0}]
a) We consider here $\delta\geq0$. We denote by $\supp(X)$ the
support of the law of a random variable $X$. We take $\delta_0$
small enough. We consider for $x\in[\log\ep';0]$ the function
\[\phi(x)=\max\{y;y\in\supp(h_n(x)-x)\},\]
which by (H2) and (\ref{hn}) is a continuous function of $x$.
Moreover, since $\E(h_n(x)-x)=0$, $\phi$ is strictly positive. By
compactness, there exists $c>0$ such that for $x\in[\log\ep';0]$,
\[\phi(x)>2c,\]
\[\P(h_n(x)\geq x+c)>0.\]
By (H2) and (\ref{hn}), $\P(h_n(x)\geq x+c)$ is continuous and
once again, by compactness, there exists $\theta>0$ such that for
$x\in[\log\ep';0]$  ,
\[\P(h_n(x)\geq x+c)>\theta.\]

b) If there are at least $N$ steps in a row such that
$h_n(w_{n-1})\geq x+c$, then the chain reaches $0$. By the point
a), that happens with probability at least $\theta^N>0$, hence $0$
is a recurrent point for the chain $w_n$ and the time between two
visits at 0 is a.s bounded above by $N$ times a geometric random
variable of parameter $\theta^N$.
\end{proof}

We continue with the proof of Lemma \ref{gros} b). By Proposition
\ref{ergo2} a), to prove that the proportion of the time that
$w_n$ spends above $\log\ep$ goes to $0$ as $\delta$ goes to $0$,
we only need to prove that
\[\nu_\delta([\log\ep,0])\xrightarrow[\delta\rightarrow0]{}0.\]
Let us first prove that $\E
T\xrightarrow[\delta\rightarrow0]{}\infty$, which by Proposition
\ref{ergo2} b) will prove that
\[\nu_\delta(0)\xrightarrow[\delta\rightarrow0]{}0.\] We use the
following lemma.

\begin{lem}
\label{technic} Assume (H2).
\begin{enumerate}
    \item There exist $u>0$ and $\alpha>0$ dependent on $\ep$ and independent of $\delta$ such that for all $x\in[2\log\ep;0]$,
      \[\P\left(h_n(x)\geq x+u\right)>\alpha.\]
  \item There exist $v>0$ and $\beta>0$ dependent on $\ep$ and independent of $\delta$ such that
      \[\P\left(\log\ep<h_1(0)<-v\right)>\beta.\]
\end{enumerate}
\end{lem}
The lemma will be proved later on.

We continue with the proof of Lemma \ref{gros} b). We denote
${\cal A}$ the event $\log\ep<h_1(0)<-v$. On ${\cal A}$, we define
the stopping time
\[\tilde{T}=1+\inf\{n\geq1 ; h_{n+1}(w_{n})>0\ \textrm{or}\ h_{n+1}(w_{n})<\log\ep'\}.\]
We now condition on the event ${\cal A}$ and on $x_0=h_1(0)$,
denote by $\tilde{\P}$ and $\tilde{\E}$ the associated probability
and expectation. $\tilde{T}\leq T$ is the first time the chain is
truncated. Moreover, for $n<\tilde{T}$, $w_n=h_n(w_{n-1})$, so
with (\ref{mart}), by classical martingale arguments,
\[\tilde{\E}\left(h_{\tilde{T}}(w_{\tilde{T}-1})\right)=x_0.\]

We denote by ${\cal A}_0$ the event that $w_n$ reaches $\log\ep'$
before $0$, we set $p=\tilde{\P}({\cal A}_0)$,
$X_0=\tilde{\E}\left(h_{\tilde{T}}(w_{\tilde{T}-1})|{\cal
A}_0^c\right)$ and
$X_1=\tilde{\E}\left(h_{\tilde{T}}(w_{\tilde{T}-1})|{\cal
A}_0\right)$.
\[x_0=pX_1+(1-p)X_0.\]
\[p=\frac{X_0-x_0}{X_0-X_1}.\]
$X_0\geq0$ and $X_1\leq\log\ep'\leq\log\ep<x_0<-v$ hence,
\begin{align}
\label{situveux} p\geq\frac{-x_0}{-X_1}\geq\frac{v}{-X_1}.
\end{align}
Using $X_1\leq \log\ep'$, (\ref{situveux}) and $w_n^2-Vn$, which
is a super-martingale by Lemma \ref{jump} b),
\begin{align*}
\tilde{\E}(T)\geq\tilde{\E}(\tilde{T})&\geq\frac{\tilde{\E}\left(h_{\tilde{T}}(w_{\tilde{T}-1})\right)^2-x_0^2}{V}\\
&\geq\frac{pX_1^2+(1-p)X_0^2-x_0^2}{V}\\
&\geq\frac{pX_1^2-x_0^2}{V}\\
&\geq\frac{v(-X_1)-x_0^2}{V}\\
&\geq\frac{v(-\log\ep')-x_0^2}{V}.\\
\end{align*}
We integrate over $x_0$ and use $\P({\cal A})>\beta$ and
$\E(h_1(0)^2|{\cal A})<(\log\ep)^2$.
\[\E(T|{\cal A})\geq\frac{v(-\log\ep')-\E(h_1(0)^2|{\cal A})}{V}.\]
\[\E(T)\geq\beta\frac{v(-\log\ep')-(\log\ep)^2}{V}.\]

We have proved that $\E
T\xrightarrow[\delta\rightarrow0]{}\infty$, which proves that
$\nu_\delta(0)\xrightarrow[\delta\rightarrow0]{}0$.

Using Lemma \ref{technic} and the invariance of $\nu_\delta$, let
us prove by induction that for ${\cal
N}\leq\left\lceil\frac{-\log\ep}{u}\right\rceil$,
\[\nu_\delta\left([-{\cal N}u;0]\right)\leq \alpha^{-{\cal N}}\nu_\delta(0).\]
\begin{align*}
\nu_\delta([-({\cal N}-1)u;0])&\geq\int\nu_\delta(dw_0)\P_{w_0}(w_1\in[-({\cal N}-1)u;0])\\
&\geq\int_{[-{\cal N}u;0]}\nu_\delta(dw_0)\P_{w_0}(w_1\in[-({\cal N}-1)u;0])\\
&\geq\int_{[-{\cal N}u;0]}\nu_\delta(dw_0)\P_{w_0}(h_1(w_0)\geq u+w_0)\\
&\geq\alpha\nu_\delta([-{\cal N}u;0]).
\end{align*}

Therefore,
\[\nu_\delta\left([\log\ep;0]\right)\leq\alpha^{\left\lceil\frac{-\log\ep}{u}\right\rceil}\nu_\delta(0).\]
So,
\[\nu_\delta\left([\log\ep;0]\right)\xrightarrow[\delta\rightarrow0]{}0.\]
That concludes the proof of Lemma \ref{gros} b).

\begin{proof}[Proof of Lemma \ref{technic}]
We consider here $\delta\geq0$. We denote by $\supp(X)$ the
support of the law of a random variable $X$. We take $\delta_0$
small enough.

a) We consider for $x\in[2\log\ep;0]$ and
$0\leq\delta\leq\delta_0$ the function
\[\phi(x,\delta)=\max\{y;y\in\supp(h_n(x)-x)\},\]
which by (H2) is a continuous function of $(x,\delta)$ because
$(h_n(x)-x)$ is continuous in $(x,\delta)$. Moreover, since
$\E(h_n(x)-x)=0$, $\phi$ is strictly positive. By compactness,
there exists $u>0$ such that for $x\in[2\log\ep;0]$ and
$0\leq\delta\leq\delta_0$,
\[\phi(x,\delta)>2u,\]
\[\P(h_n(x)\geq x+u)>0.\]
By (H2), $\P(h_n(x)\geq x+u)$ is continuous and once again, by
compactness, there exists $\alpha>0$ such that for
$x\in[2\log\ep;0]$ and $0\leq\delta\leq\delta_0$,
\[\P(h_n(x)\geq x+u)>\alpha.\]
b) For all $0\leq\delta\leq \delta_0$, there exist $\ep_0>0$ and
$v>0$ such that $\P(\log\ep_0<h_1(0)<-v)>0$. Like in the proof of
a, by (H2), we can chose $\ep_0>0$ and $v>0$ continuous in
$\delta$. By compactness, we can chose $\ep_0>0$ and $v>0$
independent of $\delta$ such that for all $0\leq\delta\leq
\delta_0$, $\P(\log\ep_0<h_1(0)<-v)>0$ and like in the proof of
a), by (H2), that probability can be chosen continuous in
$\delta$. Therefore, by compactness again, there exists $\beta>0$
dependent on $\ep$ and independent of $\delta$ such that
$\P(\log\ep_0<h_1(0)<-v)>\beta$. Take $\ep<\ep_0$, we have
$\P(\log\ep<h_1(0)<-v)>\beta$.
\end{proof}

\begin{proof}[Proof of Lemma \ref{jump}]
Note that by (H1), $V<\infty$. $j_n(x)$ is a non-increasing
continuous function of $x$ and so is $\E j_n(x)$. $\E j_n(0)<0$,
hence given $\delta$, there exist $0<\ep'<1$ such that $\E
j_n(\log\ep')\leq0$, and for $x\geq\log\ep'$, $\E j_n(x)\leq0$.
That gives point 1. For point 2, take $C$ such that for all
$x\geq0$,
\[(\log(x+1))^2\leq(\log(x))^2+C.\]
To prove that $\lim_{\delta\rightarrow0}\ep'=0$, it is enough to
prove that for a given $L<0$, we can find $\delta$ small enough
such that $\E j_n(L)\leq0$. That is true because for a given $L$,
$\E j_n(x)$ is a continuous function of $\delta$ which, by (H4) is
negative for $\delta=0$.
\end{proof}

\subsection{Proof of Theorem \ref{principalKtext}}
\label{sec:KUsersModel}

We reformulate the problem in the spirit of Appendix
\ref{sec:DefinitionsAndMainResult}. Let $K>1$. The $a_i$ (resp.
$b_i$) are now independent complex vectors of size $K$ whose
coefficients are independent and distributed according to $\pi_a$
(resp. $\pi_b$). We denote by $\P$ the probability associated with
those random sequences and by $\E$ the associated expectation. We
consider the following $M\times K(M+1)$ channel transfer matrix:
\[\Mat{H}_M=
\begin{pmatrix}
\vct{a}_1&\vct{b}_1&0&\cdots&0\\
0&\ddots&\ddots&\ddots&\vdots\\
\vdots&\ddots&\ddots&\ddots&0\\
0&\cdots&0&\vct{a}_M&\vct{b}_M\\
\end{pmatrix}.\]

We consider the following variable
\[\mathcal{C}_{M}(P)=\frac{1}{M}\tr\left\{\log\left(I+\frac{P}{K}\Mat{H}_M\Mat{H}_M^\dagger\right)\right\},\]
where $P=K\rho$. Note that,
\[\Mat{H}_M\Mat{H}_M^\dagger=\begin{pmatrix}
\abs{\vct{a}_1}^2+\abs{\vct{b}_1}^2&<\vct{a}_2;\vct{b}_1>&0&\cdots&0\\
<\vct{a}_2;\vct{b}_1>^\dagger&\abs{\vct{a}_2}^2+\abs{\vct{b}_2}^2&<\vct{a}_3;\vct{b}_2>&\ddots&\vdots\\
0&\ddots&\ddots&\ddots&0\\
\vdots&\ddots&\ddots&\ddots&<\vct{a}_M;\vct{b}_{M-1}>\\
0&\cdots&0&<\vct{a}_M;\vct{b}_{M-1}>^\dagger&\abs{\vct{a}_M}^2+\abs{\vct{b}_M}^2\\
\end{pmatrix},\]
where $\abs{\vct{a}_i}^2=\sum_{k=1}^K\abs{a_{i,k}}^2$ and
$<\vct{a}_i,\vct{b}_j>=\sum_{k=1}^K (a_{i,k})^\dagger b_{j,k}$.

\begin{thm}
\label{principalK} Assume (H1), (H2) and (H4)
\begin{enumerate}
    \item For every $\rho>0$, $\mathcal{C}_M(P)$
        converges $\P$-a.s as $M$ goes to infinity. We call the limit $\mathcal{C}(P)$.
    \item As $P$ goes to infinity,
        \[\mathcal{C}(P)=\log P+
        \E\log\left(\frac{e+\abs{\vct{b}}^2}{K}\right)+o(1),\]
    where the expectation is taken in the following way. $e$ and $b$ are independent. $\vct{b}$ is a complex $K$-vector whose coefficients are independent and distributed according to $\pi_b$. The law of $e$ is $m_0$, which is the unique invariant probability of the Markov chain defined by
    \[e_{n+1}=\abs{\vct{a}_n}^2\left(\frac{e_n+\abs{\vct{b}_{n-1}}^2\sin^2(\vct{a}_n,\vct{b}_{n-1})}{e_n+\abs{\vct{b}_{n-1}}^2}\right).\]
\end{enumerate}
\end{thm}

The rest of this appendix is devoted to the proof of Theorem
\ref{principalK}.

As in Appendix \ref{sec:DefinitionsAndMainResult}, we define the
sequence $x_n$ as follows. $x_0=0$, $x_1=1$, and for $n\geq1$,
\begin{align}
\label{xn1K}
x_{n+1}=\frac{\lambda-\abs{\vct{a}_n}^2-\abs{\vct{b}_n}^2}{<\vct{a}_{n+1};\vct{b}_{n}>}x_n-\frac{<\vct{a}_n;\vct{b}_{n-1}>^\dagger}{<\vct{a}_{n+1};\vct{b}_{n}>}x_{n-1}.
\end{align}
We get, like in (\ref{xn}), for $\lambda=-1/\rho$,
\begin{align}
\label{xnK}
\mathcal{C}_M(P)=\log(P/K)+\frac{1}{M}\log\abs{x_{M+1}(\lambda)}+\frac{1}{M}\sum_{i=1}^M\log\abs{<\vct{a}_{i+1};\vct{b}_i>}\
\ \ \ \ \ \ \P-\textrm{a.s}.
\end{align}
Set $c_n\triangleq x_n/x_{n-1}$, for $n\geq2$. By (\ref{xn1K}), we
get
\[c_{n+1}=\frac{\lambda-\abs{\vct{a}_n}^2-\abs{\vct{b}_n}^2}{<\vct{a}_{n+1};\vct{b}_{n}>}-\frac{<\vct{a}_n;\vct{b}_{n-1}>^\dagger}{c_n<\vct{a}_{n+1};\vct{b}_{n}>}.\]
Let $d_n=c_n<\vct{a}_n;\vct{b}_{n-1}>$. Then,
\begin{align*}
d_{n+1}&=\lambda-\abs{\vct{a}_n}^2-\abs{\vct{b}_n}^2-\frac{\abs{<\vct{a}_n;\vct{b}_{n-1}>}^2}{d_n}\\
&=\lambda-\abs{\vct{b}_n}^2-\abs{\vct{a}_n}^2\left(1+\frac{\abs{\vct{b}_{n-1}}^2\cos^2(\vct{a}_n,\vct{b}_{n-1})}{d_n}\right),
\end{align*}
where
\[\cos^2(\vct{a}_n,\vct{b}_{n-1})\triangleq\abs{<\vct{a}_n;\vct{b}_{n-1}>}^2/\abs{\vct{a}_n}^2\abs{\vct{b}_{n-1}}^2.\]

Note that $0\leq\cos^2\leq1$. Let
$e_{n}=-d_n-\abs{\vct{b}_{n-1}}^2$.
\begin{align}
\label{enK}
e_{n+1}=-\lambda+\abs{\vct{a}_n}^2\left(\frac{e_n+\abs{\vct{b}_{n-1}}^2\sin^2(\vct{a}_n,\vct{b}_{n-1})}{e_n+\abs{\vct{b}_{n-1}}^2}\right),
\end{align}
where $\sin^2\triangleq 1-\cos^2$. With the initial conditions,
$d_2<-\abs{\vct{b}_1}^2$, hence $e_2>0$ and for all $n$, $e_n>0$.
Note that $(e_n)$ is a Markov chain and that for all $n$, $e_n$ is
independent of $\vct{a}_n$ and $\vct{b}_{n-1}$. By (\ref{xnK}), we
get

\begin{equation}
\label{goalK}
\begin{split}
\mathcal{C}_M(P)
&=\log(P/K)+\frac{1}{M}\sum_{i=2}^{M+1}\log\abs{c_i(\lambda)}+\frac{1}{M}\sum_{i=1}^M\log\abs{<\vct{a}_{i+1};\vct{b}_i>}\\
&=\log(P/K)+\frac{1}{M}\sum_{i=2}^{M+1}\log(\abs{d_i})+o(1)\\
&=\log(P)+\frac{1}{M}\sum_{i=2}^{M+1}\log\left(\frac{e_i(\lambda)+\abs{\vct{b}_{i-1}}^2}{K}\right)+o(1)\\
\end{split}
\end{equation}

We only need to study the Markov chain $(e_n,\vct{b}_{n-1})$. For
convenience, we set $\delta=-\lambda$ and we allow $\delta=0$. We
also assume without loss of generality that the chain starts at
$(e_1,\vct{b}_0)$.

\begin{prop}
\label{ergoK} Assume (H2) and (H4). Take $\delta\geq0$. The Markov
chain $(e_n(\delta),\vct{b}_{n-1})$ has a unique stationary
probability, say, $\mu_\delta$ and for $s\in L^1(\mu_\delta)$, for
every starting point $(e_1,\vct{b}_0)\in\R_+\times\C^K$,
$\P_{(e_1,\vct{b}_0)}$-a.s,
\[\frac{1}{n}\sum_{i=0}^{n}s(e_i,\vct{b}_{i-1})\tendvers\int s d\mu_\delta.\]
Moreover, $\mu_\delta$ is weakly continuous in $\delta=0$.
\end{prop}
\begin{proof}
We consider the Markov chain $(e_n)$ on the compact $[0,\infty]$.
By (\ref{enK}), for $n\geq1$ and $e\in[0,\infty]$,
$\P_e(e_n=\infty)=0$. Consider (\ref{enK}), by (H2), for
$e_1\in[0,\infty)$, the law of $e_2$ under $P_{e_1}$ is absolutely
continuous with respect to the Lebesgue measure on
$[\delta,\infty]$. Moreover, by (H4), the law of $e_2$ under
$P_{e_1}$ and the Lebesgue measure on $[\delta,\infty]$ are
mutually absolutely continuous. Therefore, for $e_1\in[0,\infty)$
and $n\geq3$, the law of $e_n$ under $P_{e_1}$ and the Lebesgue
measure on $[\delta,\infty]$ are mutually absolutely continuous.
That fact allows us to prove like in Appendix
\ref{sec:StudyOfTheMarkovChain} that the Markov chain $(e_n)$ is
$l$-irreducible, positive Harris with invariant probability
$m_\delta$, where $l$ is the Lebesgue measure on
$[\delta,\infty]$. Since $\P_e(e_n=\infty)=0$, $m_\delta$ does not
charge $\{\infty\}$. We identify $m_\delta$ and the measure it
induces on $\R_+$. We denote by $\Pi_{\vct{b}}$ the law of
$\vct{b}$. Since for $n\geq1$, $e_n$ and $\vct{b}_{n-1}$ are
independent, the Markov chain $(e_n,\vct{b}_{n-1})$ is
$l\times\Pi_{\vct{b}}$-irreducible, positive Harris with invariant
probability $\mu_\delta=m_\delta\times\Pi_{\vct{b}}$. By
\cite[Theorem 17.0.1]{harris}, the Markov chain
$(e_n(\lambda),\vct{b}_{n-1})$ has a unique stationary probability
$\mu_\delta$ and for $s\in L^1(\mu_\delta)$, for every starting
point $(e_1,\vct{b}_0)\in\R_+\times\C^K$,
$\P_{(e_1,\vct{b}_0)}$-a.s,
\[\frac{1}{n}\sum_{i=0}^{n}s(e_i,\vct{b}_{i-1})\tendvers\int s d\mu_\delta.\]

Let us prove that $m_\delta$ converges weakly to $m_0$ when
$\delta$ converges to 0, which will finish the proof.
$\{m_\delta,\delta\geq 0\}$ are measures on the compact
$[0,\infty]$ hence it is enough to show that $m_0$ is the only
limit point when $\delta$ goes to 0. By (H2), for a point $x$ and
an interval $A$ in $[0,\infty]$, $\P_{e_1}(e_2(\delta)\in A)$
converges to $\P_{e_1}(e_2(0)\in A)$. It implies that a limit
point must be an invariant measure for the chain with $\delta=0$.
The only possibility is $m_0$.
\end{proof}

By (\ref{enK}), $m_\delta$ is stochastically dominated by the law
of $\abs{\vct{a}_n}^2+\delta$. Therefore, by (H1),
$(x,y)\rightarrow\log(x+y)\in L^1(\mu_\delta)$. (\ref{goalK}) and
Proposition \ref{ergoK} conclude the proof of Theorem
\ref{principalK}.

\subsection{Product of
random matrices}\label{sec:ProductOfRandomMatrices}

We prove Lemma \ref{gros} a) assuming only (H1) and (H2). We use
the theory of product of random matrices theory. For a general
introduction to the aspects of the theory we use here, the reader
may consult \cite{enorme}, \cite{groslivre}, \cite{petitlivre}-\nocite{oseledec}\cite{petit}.

Let us take $\abs{\cdot}$ any norm on $\C^2$ and $\nrm{\cdot}$ the
associated operator norm on matrices. For a given $\lambda$,
\[\begin{pmatrix}x_{n+1}\\x_n\end{pmatrix}=
\begin{pmatrix}
\frac{\lambda-\abs{a_n}^2-\abs{b_n}^2}{a_{n+1}^\dagger b_{n}}&-\frac{a_n b_{n-1}^\dagger}{a_{n+1}^\dagger b_{n}}\\
1&0
\end{pmatrix}
\begin{pmatrix}x_{n}\\x_{n-1}\end{pmatrix}\]
For $a,a',b,b'\in\C-0$, we define the following invertible matrices
\[\Mat{g}(\lambda,a,a',b,b')\triangleq\begin{pmatrix}
\frac{\lambda-\abs{a}^2-\abs{b'}^2}{a'^\dagger b'}&-\frac{ab^\dagger }{a'^\dagger b'}\\
1&0
\end{pmatrix}.\]
Finally, we define
\[\Mat{g}_n(\lambda)\triangleq \Mat{}g(\lambda,a_n,a_{n-1},b_{n-1},b_n)=\begin{pmatrix}
\frac{\lambda-\abs{a_n}^2-\abs{b_n}^2}{a_{n+1}^\dagger b_{n}}&-\frac{a_n b_{n-1}^\dagger}{a_{n+1}^\dagger b_{n}}\\
1&0
\end{pmatrix},\]
\[\Mat{M}_n\triangleq \Mat{g}_{n}\ldots \Mat{g}_1.\]
So that
\[\begin{pmatrix}x_{n+1}\\x_n\end{pmatrix}=\Mat{M}_n\begin{pmatrix}1\\0\end{pmatrix}.\]

Set ${\cal E}=(\C-{0})^4$ which is a Borel set of a separable and
complete metric space.
$\vct{X}_n\triangleq(a_{n+1},a_n,b_n,b_{n-1})$ is a Markov chain
on ${\cal E}$, with invariant measure
$\Pi\triangleq\pi_a\times\pi_a\times\pi_b\times\pi_b$. With (H1),
\[\E_{\Pi}\left(\log^+\nrm{\Mat{g}(\lambda,a,a',b,b')}+\log^+\nrm{{\Mat{g}(\lambda,a,a',b,b')}^{-1}}\right)<\infty.\]
Notice that $\Mat{g}_n(\lambda)$ is a continuous function of
$\vct{X}_{n}$, therefore $((\vct{X}_n,\Mat{M}_n),\Pi)$ is a
multiplicative Markovian process. By \cite[Example 1 and
Proposition 2.5]{bougerol}, $1/n\log\nrm{\Mat{M}_n(\lambda)}$
converges $\P$-almost surely and in $\L_1(\Omega)$, we set
\begin{align}
\label{fur}
\gamma(\lambda)=\lim_{n\rightarrow\infty}\frac{1}{n}\log\nrm{\Mat{M}_n(\lambda)}.
\end{align}

$\gamma(\lambda)$ is the first Lyapunov exponent.

The $\L_1(\Omega)$ convergence already gives an easy upper bound
for $\gamma(\lambda)$. By the property of operator norm,
\[\gamma(\lambda)\leq\E_{\Pi}\log\nrm{\Mat{g}_1(\lambda)}.\]
Moreover, we can refine that bound into a whole family of upper
bounds, for $k\in\N$,
\begin{align}
\label{upperbound}
\gamma(\lambda)\leq\frac{1}{k}\E_{\Pi}\log\nrm{\Mat{g}_1(\lambda)...\Mat{g}_{k}(\lambda)}.
\end{align}
Note that this upper bound is getting better as $k$ increases and
tight as $k\rightarrow\infty$.

Let us now prove that
\[\frac{1}{n}\log\abs{x_{n+1}(\lambda)}\tendvers\gamma(\lambda).\]

\begin{deft}
The multiplicative system $((\vct{X}_n,\Mat{M}_n),\Pi)$ is
irreducible if there is no measurable non-random family
$\{V(\vct{X}),\vct{X}\in E\}$ of proper subspaces of $\C^2$ such
that
\[\Mat{M}_nV(\vct{X}_0)=V(\vct{X}_n),\ \ \ \ \ \P\textrm{-a.s},\ \forall n\in\N.\]
\end{deft}

\begin{lem}
Assume (H2). The multiplicative system
$((\vct{X}_n,\Mat{M}_n),\Pi)$ is irreducible
\end{lem}

The proof is an adaptation of the proof of \cite[Proposition
6.1.1]{bougerol2}.

\begin{proof}
The proof is by contradiction. Assume that there is a measurable
family $\{V(\vct{X}),\vct{X}\in E\}$ of proper subspaces of $\C^2$
such that
\[\Mat{g}_3V(\vct{X}_2)=V(\vct{X}_3),\ \ \ \ \ \P\textrm{-a.s.},\ \forall n\in\N.\]
We parameterize the proper subspaces of $\C^2$ by
$\begin{pmatrix}c\\1\end{pmatrix}$ for $c$ in $(-\infty,\infty]$.
There is a measurable family $\{c(\vct{X}),\vct{X}\in E\}$ such
that $\Mat{g}_3\begin{pmatrix}c(\vct{X}_2)\\1\end{pmatrix}$ and
$\begin{pmatrix}c(\vct{X}_3)\\1\end{pmatrix}$ are
$\P\textrm{-a.s.}$ collinear. A direct computation gives
\[c(a_4,a_3,b_3,b_2)=\frac{\lambda-\abs{a_3}^2-\abs{b_3}^2}{a_4^\dagger b_3}-\frac{a_3 b_2^\dagger}{c(a_3,a_2,b_2,b_1)a_4^\dagger b_3},\ \ \ \ \ \P\textrm{-a.s.},\]
that is
\[c(a_3,a_2,b_2,b_1)=\frac{a_3 b_2^\dagger }{a_4^\dagger b_3\left(\frac{\lambda-\abs{a_3}^2-\abs{b_3}^2}{a_4^\dagger b_3}-c(a_4,a_3,b_3,b_2)\right)},\ \ \ \ \ \P\textrm{-a.s.}.\]
Note that the RHS does not depend on $a_2$ and $b_1$, hence,
$c(a,a',b,b')$ does not depend on $a'$ and $b'$. Setting
$d(a,b)=a^\dagger b\;c(a,b)$, we get
\begin{align}
\label{d}
d(a_4,b_3)=\lambda-\abs{a_3}^2-\abs{b_3}^2-\frac{\abs{a_3}^2\abs{b_2}^2}{d(a_3,b_2)},\
\ \ \ \ \P\textrm{-a.s.}.
\end{align}
The RHS does not depend on $a_4$, hence, $d(a,b)$ does not depend
on $a$. From (\ref{d}), we get
\[\frac{d(b_2)}{\abs{b_2}^2}=-\frac{{\abs{a_3}^2}}{d(b_3)-\lambda+\abs{a_3}^2+\abs{b_3}^2},\ \ \ \ \ \P_{\pi}\textrm{-a.s.}.\]
The RHS does not depend on $b_2$, hence, $d(b)/\abs{b}^2$ does not
depend on $b$, set $d(b)=L\abs{b}^2$, where $L$ is a fixed
constant. Then,
\[(L+1)\abs{b_3}^2=\lambda-\abs{a_3}^2\left(1+\frac{1}{L}\right),\ \ \ \ \ \P_{\pi}\textrm{-a.s.}.\]
If $L\neq-1$, $\abs{b_3}^2$ is a measurable function of $a_3$ and
since it is also independent of $a_3$, it is a constant, which is
in contradiction with (H2). Hence $L=-1$, which gives a
contradiction with $\lambda<0$.
\end{proof}

By \cite[Lemma 2.6]{bougerol}, irreducibility implies that
\[\lim_{n\rightarrow\infty}\frac{1}{n}\log\abs{\begin{pmatrix}x_{n+2}\\x_{n+1}\end{pmatrix}}=\gamma\ \ \ \ \P-\textrm{a.s}.\]
The following lemma completes the proof.
\begin{lem}
Assume (H1). \label{bc}
\[\lim_{n\rightarrow\infty}\frac{1}{n}\left(\log\abs{\begin{pmatrix}x_{n+2}\\x_{n+1}\end{pmatrix}}-\log\abs{x_{n+1}}\right)=0\ \ \ \ \P-\textrm{a.s}.\]
\end{lem}

\begin{proof}
\[\log\abs{\begin{pmatrix}x_{n+2}\\x_{n+1}\end{pmatrix}}-\log\abs{x_{n+1}}=\log\abs{\begin{pmatrix}c_{n+2}\\1\end{pmatrix}}\geq0.\]
Let us prove that for $\ep>0$,
$\P\left(\frac{1}{n}\log\abs{\begin{pmatrix}c_n\\1\end{pmatrix}}\geq\ep\right)$
is a summable series, which by the Borel-Cantelli Lemma will prove
the lemma. We have
\begin{align}
\label{eq-080807a}
\P\left(\frac{1}{n}\log\abs{\begin{pmatrix}c_n\\1\end{pmatrix}}\geq\ep\right)
&\leq\P\left(\frac{1}{n}\log\left(\abs{c_n}+1\right)\geq\ep\right)\nonumber\\
&\leq\P\left(\abs{c_n}\geq e^{n\ep}-1\right)\nonumber\\
&\leq\P\left(\abs{c_n}\geq e^{\frac{n\ep}{2}}\right)\nonumber\\
&\leq\P\left(\frac{\abs{b_{n-1}}}
{\abs{a_n}}\frac{1}{1-e_n}\geq e^{\frac{n\ep}{2}}\right)\nonumber\\
&\leq\P\left(\frac{\abs{b_{n-1}}}{\abs{a_n}}\geq
e^{\frac{n\ep}{4}}\right)+ \P\left(\frac{1}{1-e_n}\geq
e^{\frac{n\ep}{4}}\right).
\end{align}
We analyze the  right side of (\ref{eq-080807a}). We use the fact
that $\log\abs{a_n}$ and $\log\abs{b_{n-1}}$ have a second moment
by (H1) and that it does not depend on $n$. By the
Bienaym\'e-Tchebicheff inequality, we get
\begin{equation}\label{term}\begin{split}
\P\left(\frac{\abs{b_{n-1}}}{\abs{a_n}}\geq e^{\frac{n\ep}{4}}\right)&=\P\left(\log\abs{b_{n-1}}-\log\abs{a_n}\geq \frac{n\ep}{4}\right)\\
&\leq\frac{16\E\left(\left(\log\abs{b_{n-1}}-
\log\abs{a_n}\right)^2\right)}{n^2\ep^2},
\end{split}\end{equation}
implying that the first term in the right side of
(\ref{eq-080807a}) forms a summable series. Moreover
\[\log\frac{1}{1-e_n}\leq\log\frac{-\lambda+
\abs{b_{n-1}}^2+\abs{a_{n-1}}^2}{\abs{b_{n-1}}^2},\] which has a
second moment by (H1), hence, by a computation like (\ref{term})
and the Bienaym\'e-Tchebicheff inequality,
$\P\left(\frac{1}{1-e_n}\geq e^{\frac{n\ep}{4}}\right)$ is a
summable series. The Borel-Cantelli Lemma applied to the right
side of (\ref{eq-080807a}) concludes the proof.
\end{proof}

\subsection{Determinants of Jacobi Matrices}
\label{app: Jacobi det} An interesting and useful characterization
of an $M\times M$ Jacobi matrix is that its determinant can be
expressed by the following recursive formula \cite{Horn-1985}
\begin{equation}\label{eq: Narula recursive formula}
    \det \Mat{G}_m = [\Mat{G}_m]_{m,m}\det
    \Mat{G}_{m-1}-[\Mat{G}_m]_{m,m-1}[\Mat{G}_m]_{m-1,m}\det\Mat{G}_{m-2}\quad;\quad
    m=3,\ldots,M\ ,
\end{equation}
with
\begin{equation}\label{eq: Narula initial conditions}
\begin{aligned}
    \det \Mat{G}_1 &= [\Mat{G}_m]_{1,1}\\
    \det \Mat{G}_2 &=
    [\Mat{G}_m]_{1,1}[\Mat{G}_m]_{2,2}-[\Mat{G}_m]_{1,2}[\Mat{G}_m]_{2,1}\
    ,
\end{aligned}
\end{equation}
where $\Mat{G}_{m}$ is the principle submatrix of $\Mat{G}_M$,
obtained by deleting its last $(M-m)$ columns. This
characterization already used by Narula \cite{Narula-1997}, can be
easily proved by straight forward calculation of the determinant
of $\Mat{G}_M$, starting from its last row.

Examining \eqref{eq: Narula recursive formula}, it is observed
that the determinant of a square Jacobi matrix is dependent on a
weighted sum of its two largest principle matrices' determinants
only. Furthermore, $\det\Mat{G}_{m-1}$ and $\det\Mat{G}_{m-2}$ are
independent of the entries $[\Mat{G}_m]_{m,m}$,
$[\Mat{G}_m]_{m,m-1}$, and $[\Mat{G}_m]_{m+1,m}$.

It is worth mentioning that this approach can not be extended for
matrices with a number of non-zero diagonal higher than 3. Hence,
a similar formula, can not be obtained even for five-diagonal
matrices and the resulting formula involves $O(M)$ determinants of
submatrices (not necessarily principle submatrices).

\bibliographystyle{ieeetr}
\bibliography{Reference_List}

\newpage

\begin{figure}
\begin{center}
\psfrag{a1k\r}{\scriptsize$a_{1,k}$}
\psfrag{a2k\r}{\scriptsize$a_{2,k}$}
\psfrag{a3k\r}{\scriptsize$a_{3,k}$}
\psfrag{b1k\r}{\scriptsize$b_{1,k}$}
\psfrag{b2k\r}{\scriptsize$b_{2,k}$}
\psfrag{b3k\r}{\scriptsize$b_{3,k}$}
\includegraphics[scale=0.9]{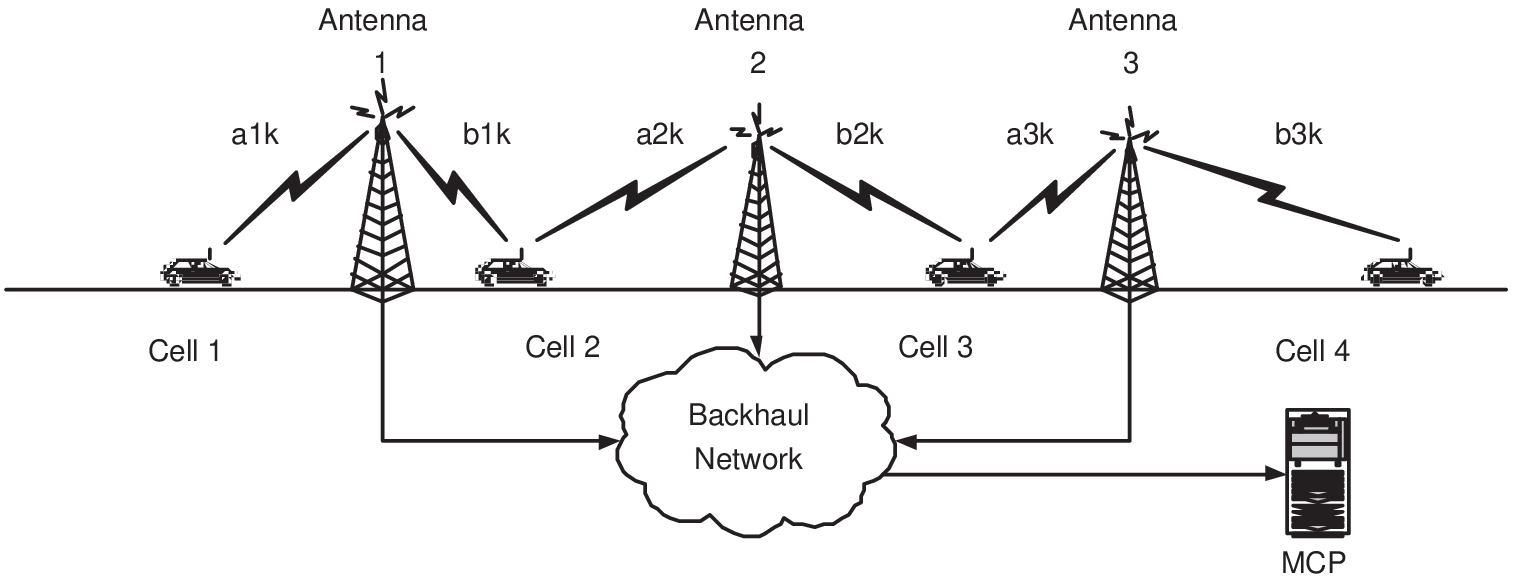}
\end{center}
\caption{Soft-Handoff setup ($M=3$)}
\label{fig: Soft-handoff
setup}
\end{figure}

\begin{figure}
\begin{center}
\includegraphics[scale=0.69]{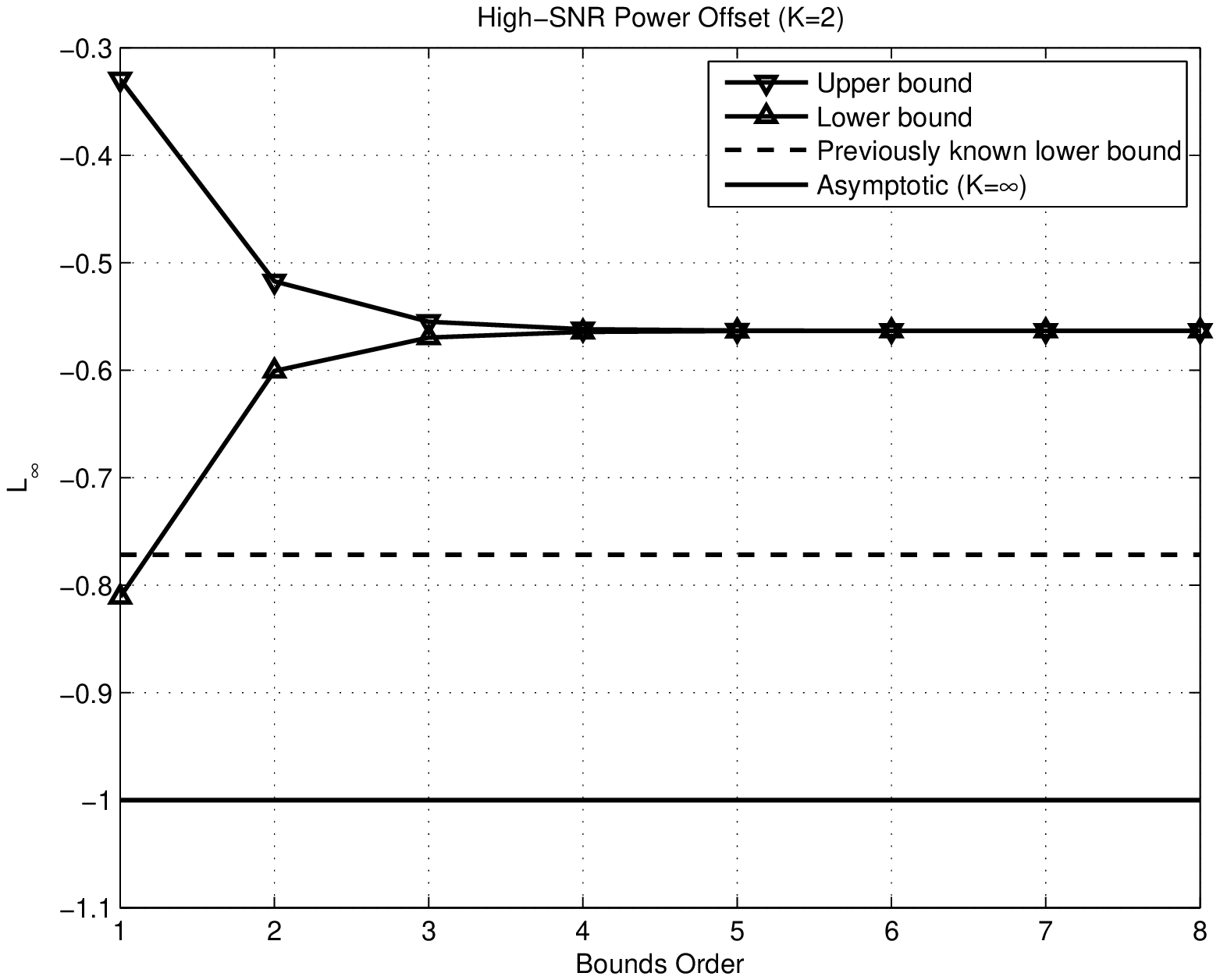}
\end{center}
\caption{High-SNR power offset bounds for Rayleigh fading, $K=2$,
and bounds order $n=1,2\cdots,8$} \label{fig: Linf order K=2}
\end{figure}

\begin{figure}
\begin{center}
\includegraphics[scale=0.65]{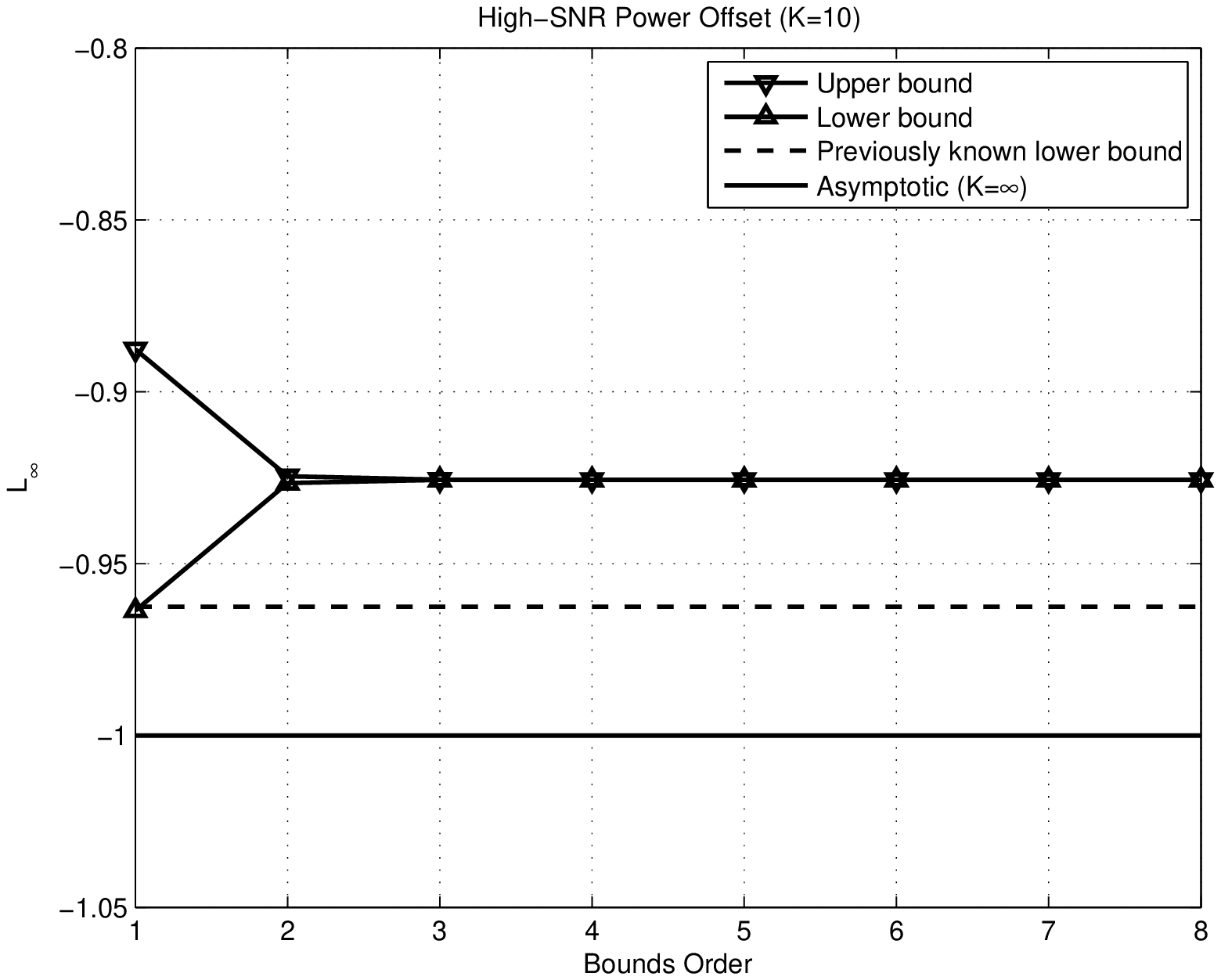}
\end{center}
\caption{High-SNR power offset bounds for Rayleigh fading, $K=10$,
and bounds order $n=1,2\cdots,8$} \label{fig: Linf order K=10}
\end{figure}

\begin{figure}
\begin{center}
\includegraphics[scale=0.65]{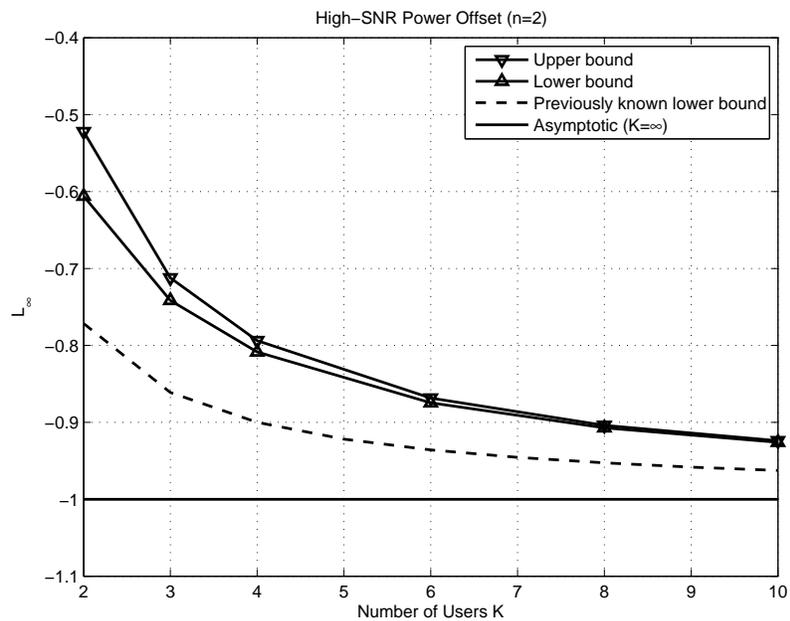}
\end{center}
\caption{High-SNR power offset bounds (order $n=2$) for Rayleigh,
and $K=2,3,4,6,8,10$. Note that for $K=1$, ${\cal
L}_\infty=\frac{\gamma}{\log 2}\approx 0.833$.} \label{fig: Linf
users n=2}
\end{figure}

\end{document}